\documentclass{article}
\usepackage[toc]{appendix}

\usepackage{PRIMEarxiv}
\RequirePackage{amsthm,amsmath,amsfonts,amssymb}
\RequirePackage[authoryear]{natbib}
\usepackage{hyperref}
\hypersetup{colorlinks,allcolors=black}
\RequirePackage{graphicx}
\usepackage{algorithm}
\usepackage{algpseudocode}
\usepackage{bbm}
\usepackage{subcaption}
\usepackage{caption} 
\usepackage{tabularray}
\usepackage{comment}

\theoremstyle{plain}

\newtheorem{thm}{Theorem}[section]
\newtheorem{lem}[thm]{Lemma}
\newtheorem{prop}[thm]{Proposition}
\newtheorem{corollary}{Corollary}[section]
\theoremstyle{remark}

\usepackage[utf8]{inputenc} 
\usepackage[T1]{fontenc}    
\usepackage{hyperref}       
\usepackage{url}            
\usepackage{booktabs}       
\usepackage{amsfonts}       
\usepackage{nicefrac}       
\usepackage{microtype}      
\usepackage{lipsum}
\usepackage{fancyhdr}       
\usepackage{graphicx}       
\graphicspath{{media/}}     

\title{Estimation of Multivariate Discrete Hawkes Processes: \\ Application to Incident Monitoring}

\author{
  Trinnhallen Brisley \\
  The University of Edinburgh \\
  \texttt{T.Brisley@sms.ed.ac.uk} \\
   \And
  Gordon Ross \\
  The University of Edinburgh \\
  \texttt{Gordon.Ross@ed.ac.uk} \\
   \And
  Daniel Paulin \\
  The University of Edinburgh \\
  \texttt{Daniel.Paulin@ed.ac.uk} \\
  \And
  Jake Easto \\
  The National Health Service \\
  \texttt{Jake.Easto@ggc.scot.nhs.uk} \\
}

\newcommand{\Poisson}{\mathrm{Poisson}}
\newcommand{\Uniform}{\mathrm{Uniform}}

\begin{document}
\maketitle

\begin{abstract}
Hawkes processes are a class of self-exciting point processes that are used to model complex phenomena. While most applications of Hawkes processes assume that event data occurs in continuous-time, the less-studied discrete-time version of the process is more appropriate in some situations. In this work, we develop methodology for the efficient implementation of discrete Hawkes processes. We achieve this by developing efficient algorithms to evaluate the log-likelihood function and its gradient, whose computational complexity is linear in the number of events. We extend these methods to a particular form of a marked discrete multivariate Hawkes process which we use to model the occurrences of violent events within a forensic psychiatric hospital. A prominent feature of our problem, captured by a mark in our process, is the presence of an alarm system which can be heard throughout the hospital. An alarm is sounded when an event is particularly violent in nature and warrants a call for assistance from other members of staff. We conduct a detailed analysis showing that such a variant of the Hawkes process manages to outperform alternative models in terms of predictive power. Finally, we interpret our findings and describe their implications.
\end{abstract}


\section{Introduction}

In this work we study the timings of violent behavior of patients in a forensic psychiatric hospital. We attempt to understand more about the interactions between instances of violence and to build a model to accurately predict future incidents. The predictive capabilities could allow for a better design of protocols within the hospital such as more effective organisation of the daily routines inside the facility, e.g., meal times, mediation and treatment timings. There is evidence to suggest an increase in the levels of stress from such daily routines \citep{Hinsby,Meehan,Wright}. This analysis is complicated by the complex nature of the event data, which consists of event times logged from multiple different areas within the hospital, and which appears to exhibit event clustering, seasonality, and non-trivial interactions between events in different wards.

There have been some previous studies on timings of attacks within forensic psychiatric hospitals, however these studies utilised only univariate statistical methodologies. The works \citep{Beck,Barnard} give some evidence for such events occurring in 'bursts' or 'clusters' and \cite{Beck} examined the temporal patterning of violent and aggressive behavior on an inpatient psychiatric ward over a one-year period. The analysis indicated that the events occurred in a non-uniform fashion, with evidence to support the contention that periodic bursts of aggression. They use a Poisson probability table to examine if incident reports were uniformly distributed across days. Their results show that under a Poisson model, days involving more than three incident reports are highly unlikely to occur, and yet such bursts occurred frequently in the data. 

When thinking about potential seasonality of the timings of events, it is worth mentioning that many of the routines of the hospital such as meal times, mediation or treatment, might be leading to increased levels of stress thus making violent behaviour more likely \citep{Hinsby,Meehan,Wright}. There are also some articles which review the effectiveness of incident-reporting systems as a method of improving patient safety through organizational learning such as \cite{Stavropoulou}. The work by \cite{Barnard} suggest that a patient who engages in more than one disruptive incident is increasingly likely to initiate a violent act might indicate that these individuals should receive more staff intervention earlier. 

Hawkes processes, first introduced by \cite{hawkes1971spectra}, are point processes capturing the self-exciting behaviours of events that occur in clusters or bursts, where the occurrence of one event makes it more probable that further events will also occur. One of the first applications was found in the earthquake literature \cite{Earthquake} where large earthquakes are often followed by smaller tremors nearby, thus exhibiting such a self-exciting structure. Hawkes processes are successfully used in many application domains that share this self-exciting nature, such as neuroscience \citep{Neuro1,Neuro2}, terrorism \citep{white_porter, white_porter2} (in the univariate setting) and retail analytics \cite{Gordons}. 

While the Hawkes process seems well suited to modelling violent event data, there are several problems which arise. Firstly, the event times in our dataset were recorded by hand by on duty staff. Since there is uncertainty of the exact time of each event, a huge proportion of the timestamps of the events were rounded to the nearest five-minute time interval. That is, the true exact incident time is uncertain. As we can see in the following section, our dataset violates an assumption of the continuous Hawkes processes, namely that events cannot occur simultaneously. This implies that the discrete variant of the Hawkes process may be more suitable for this dataset. Compared to its continuous variant, the discrete version of the Hawkes process has been sparsely studied. One application of the (multivariate) discrete Hawkes process is \cite{Simple}. They evaluate the capability of discrete Hawkes processes to model daily mortality counts at distinct phases in the COVID-19 outbreak, indicating that such processes can describe very complex dynamic behaviour. However, as shown in Section \ref{Estimation}, when there is either a large time frame, or small discrete-time windows, the likelihood calculations needed for the model fitting procedure become computationally intractable. The data set analysed by \cite{Simple} only consisted of a few years of data recorded at daily intervals, which made model fitting substantially less expensive compared to the large time frame of our dataset where the events were observed over a nine year period and logged at five-minute frequencies, resulting in far more observation windows. Thus, we must develop more efficient methodology for fitting such processes, which has an influence on our modelling approach. In Appendix \ref{Hawkes estimation} and \ref{Hawkes forecasting}, we provide, in a similar manner as in the main body of the paper, efficient algorithms for fitting and simulating the standard non-marked discrete Hawkes process. We believe that this work may be more widely applicable and will potentially be of interest to readers. 

Further issues which separate our problem from more routine applications of Hawkes processes is the multivariate structure of the data. The hospital which we study consists of numerous wards, each of which records incidents separately. Since it is possible that violent events in one ward can trigger violent events in other wards, a multivariate process is necessary. Further, some violent events in the hospital result in an alarm being sounded, which is heard throughout the hospital and alerts everyone to the fact an event has occurred. These alarms are more likely to impact other wards compared to events which do not result in an alarm being sounded. As such, our model must be able to deal with the presence of an alarm system as a covariate/mark in the data.

The layout of this paper is as follows; Section \ref{Data} describes the hospital data which we will be modelling. Section \ref{Background} introduces the Hawkes process, and its discrete and multivariate discrete forms. In Section \ref{Model} we define the particular form of a marked discrete Hawkes process which we use to model the forensic hospital dataset and allows us to take into account the excitation effect of the alarms. Also, in this section, we introduce a strategy for regularizing the model. In later sections we will see a large increase in predictive performance with this regularization strategy. Section \ref{Estimation} presents some new methodology to allow for efficient estimation of general discrete Hawkes processes as well as for our specific model for the hospital data. We also present some theory on efficiently simulation of these models, which is necessary to obtain event forecasts. In Section \ref{Results}, we report the results from four different models applied to the hospital data. Finally, Section \ref{Conclusion} concludes with a summary of our contributions and a discussion of possible future developments. In Appendix \ref{Hawkes estimation} and \ref{Hawkes forecasting} we report algorithms for the estimation and forecasting of non-marked discrete Hawkes processes, which we believe to be more widely applicable.

\section{Hospital Incidents Data}\label{Data}

We consider a dataset consisting of event times representing violent behaviour that occurred across the eighteen areas within the hospital. These areas correspond to various wards and communal areas within the hospital. This data was recorded by employees at the hospital, anonymized by the assignment of a patient number so that no patients could be identified. Access to this anonymous dataset was provided by the National Health Service. The data covers a nine year period  between the dates 26$^{th}$ June 2012 to 30$^{th}$ September 2021. We only use data up to 1$^{st}$ January 2019 to avoid the need to model anomalous changes in behaviour due to COVID 19 which will not be relevant to normal operations. The hospital we study has four sections (i.e. distinct physical areas) each of which contains three wards with a communal area in each section, as illustrated in Figure \ref{fig:ward layout}. There is an indoor and an outdoor communal area which are available to the entire hospital. This hence means there are eighteen different locations in which violent events can occur. Table \ref{tab:no of events} provides the number of events in each section of the hospital over the entire time frame, across each area of the hospital. Excluding area C6, the numbers of events within the communal areas are very small in comparison to events within wards. Although this implies we should leave these areas out of our analysis, the presence of section and hospital wide communal areas could imply that there is still some information flow between different areas of the hospital even in the absence of an alarm being sounded. With regards to including section C6 in our model, since we're unable to know where exactly in the hospital an event in C6 occurred, we expect that the interaction between events in C6 with other sections will be inhomogeneous. On top of this, adding an additional section into our process will increase the number of parameters in our model, thus introducing identifiability and computational issues. 

A visual inspection of the event times shown in Figure \ref{fig:prelim Hawkes} suggests there  is burst like behavior in our dataset. We plot the event times for a selection of wards from each section of the hospital, given by the vertical lines. We can see quite clearly that there is clustering of events within the same ward. This could imply that a Hawkes process maybe a good model to fit to this data, due to its self-exciting nature. 

The dataset is challenging for a few reasons: firstly, the event times were recorded by hand by on duty staff in the hospital. Since there is uncertainty of the exact time of an event, due to staff having to react to an event and note down the times after it has been dealt with, the vast majority of data points were rounded to the nearest five-minute time interval. That is, the true exact incident time is uncertain. Due to the uncertainty of the exact event times, our dataset violates an assumption of the continuous Hawkes processes, namely that events cannot occur simultaneously. This is not assumed by the discrete Hawkes process. We illustrate this in Table \ref{tab:simultaneous}, which shows the the proportion of total events where there were multiple events occurring in the same five-minute interval.

  \begin{figure}
    \centering
    \includegraphics[width=12cm]{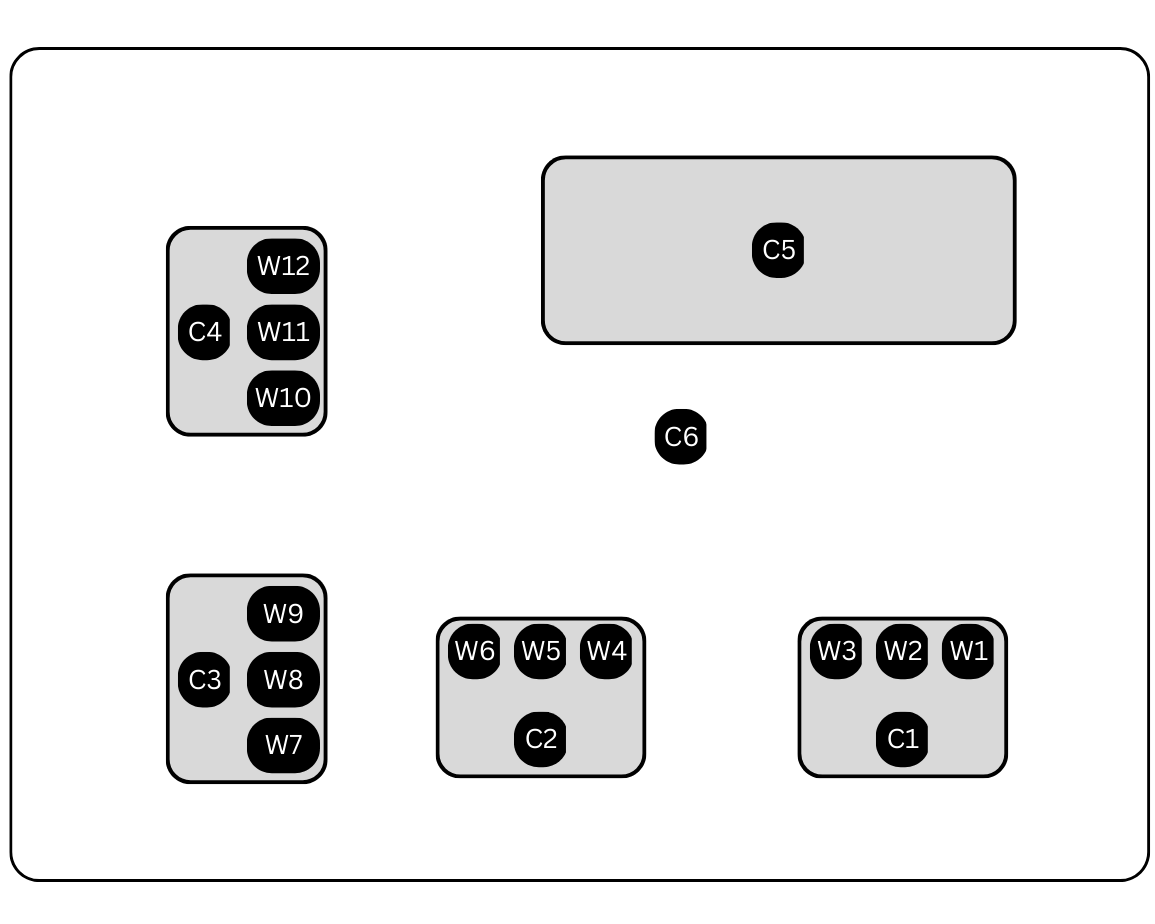}
    \caption{Ward layout: Sections W1, W2$,\ldots,$ W12 represent the wards, Sections C1, C2, C3 and C4 represent communal areas available for use by patients belonging to the wards inside the section, section C5 is a hospital-wide communal area and C6 represents the areas outside of the sections outlined previously. }
    \label{fig:ward layout}
\end{figure}

\begin{table}
\centering
\begin{tabular}{c c c c} 
\hline
Section & \multicolumn{1}{c}{\# Events}
& Section & \multicolumn{1}{c}{\# Events} \\
\hline
 C1 & 16 & W7 & 313\\
 W1 & 530 & W8 & 181\\
 W2 & 153 & W9 & 111\\
 W3 & 81 & C4 & 14\\
 C2 & 38 & W10 & 723\\
 W4 & 354 & W11 & 550\\
 W5 & 1882 & W12 & 82\\
 W6 & 247 & C5 & 14\\
 C3 & 2 &  C6 & 155\\
\hline
\end{tabular}
\caption{Total number of events recorded in each area of the hospital.}
\label{tab:no of events}
\end{table}

Secondly, with a Hawkes process \cite{hawkes1971spectra}, and with the discrete variant of which we use in this work, the occurrences of events are determined by an underlying intensity function which considers the influence from past events. The calculation of the intensity function requires time proportional to the number of time-stamps. It has been shown in continuous-time \cite{Ozaki}, and is analogously proven in discrete-time, that we can write the intensity function in terms of the cumulative mass function of the excitation kernel, defined to be a probability mass function. This allows us perform the calculation of the intensity function that requires time proportional to the number of events. Still, however, when the number of events is high, the repeated intensity function calculation will become expensive. 

We have some added information from the dataset that we use in this study. Present in the data is an indicator for whether the alarm was sounded in response to each event. The reason for an alarm being sounded is due to a particularly violent event, as determined by nearby on-shift employees, which requires urgent assistance from employees in other areas of the hospital. These alarms can be heard throughout the hospital, which implies that we need to model the added excitation due an event which triggered an alarm differently from a standard event. We can show this more concretely by performing some preliminary analysis of the events with and without alarms. More specifically, we can analyse the proportion of events that follow from an event that has/has-not caused an alarm to be sounded. One thing to note is that the effects could be due to the severity of the violent event that caused an alarm to be sounded, rather than the alarm itself. To attempt to isolate the effect of alarms, Figure \ref{tab:no of events after alarm} reports the number of events that occur in a different ward at a given amount of time after another event that has and has-not caused an alarm to be sounded. From Figure \ref{tab:no of events after alarm}, there is definitely reason to suspect quite significant effects due to the alarm being sounded rather than solely down to witnessing a particularly violent first hand.

\begin{figure}
\centering
\includegraphics[width=12cm]{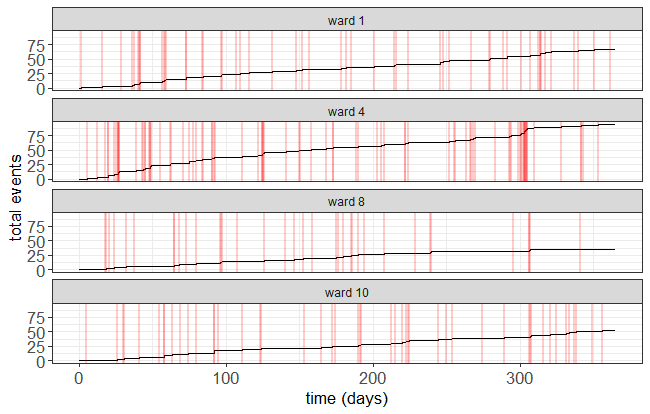}
\caption{Cumulative sum of events for a selection of wards over a one year period. Vertical lines represent the event times in the corresponding ward.}
\label{fig:prelim Hawkes}
\end{figure}

We next investigate seasonality in our dataset. There are papers which discuss seasonal characteristics in such data \citep{Beck,Peluola}. For example, \cite{Peluola} analysed a record of all incidents of violence over five years in a multilevel secure forensic hospital in Canada. There work indicated a higher occurrence of violence was recorded in the winter months compared with any other season and was related to unstructured activities. In Figure \ref{fig:prelim season}, we plot the range of event frequencies in each ward over each hour of the day along with a line representing the average we use in the model. As we can 

\begin{table}
\centering
\begin{tabular}{c c} 
\hline
N$^\textbf{o}$ simultaneous events & Proportion of total events \\
\hline
 1 & 88.78\% \\
 2 & 10.36\% \\
 3 & 0.68\% \\
 4 & 0.18\% \\
\hline
\end{tabular}
\caption{Proportion of total events where there were multiple events occurring in the same five-minute interval.}
\label{tab:simultaneous}
\end{table}

\begin{figure}[H]
\centering
\includegraphics[width=9cm]{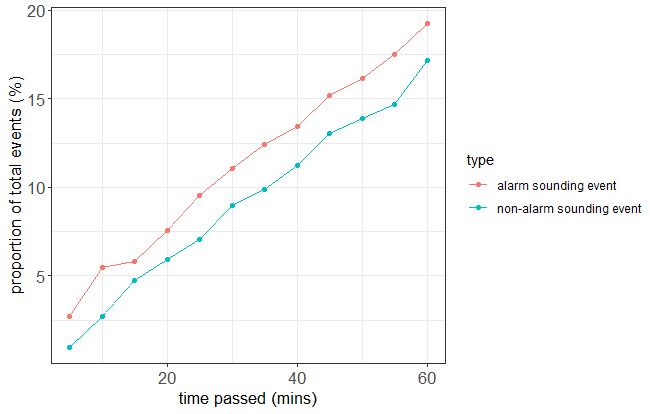}
\caption{Proportion of the total events that occur at a given amount of time immediately after another event which did/did not caused an alarm to be sounded. }
\label{tab:no of events after alarm}
\end{figure}

\begin{figure}[H]
\centering
\includegraphics[width=9cm]{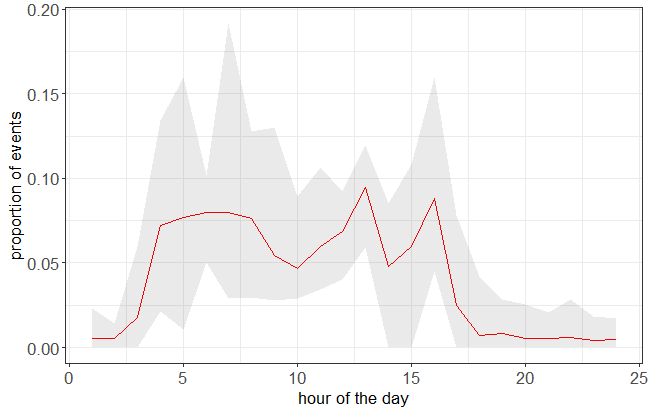}
\caption{Plot of the range of event frequencies over each ward for each our of the day with the centre line representing the sample mean. The shaded region represents the range of event frequencies over all wards in each hour of the day.}
\label{fig:prelim season}
\end{figure} 

see, their seems to be a pattern throughout the day. Namely, there are dips in activity late in the evening, during the night and during the early hours of the morning, where patients would patients would normally be in their own dorms. Note that we could not find any clear seasonal pattern across other time frames, thus have only included a time-of-day seasonal component in our model.

\section{Background}\label{Background}

Our aim is to develop a model for the number of events in each area of the forensic hospital. In order to deal with the discrete bins we use the discrete Hawkes framework, which we will introduce in this section. A marked discrete Hawkes process is an extension of the discrete Hawkes process that  describes not only the timings and locations of random events but also additional information attached to each event, known as a mark. In later sections we will use a specific form of the mark in the process to account for the additional information present in our data, that is an indicator to whether an alarm has been sounded. The remainder of this section reviews Hawkes process and its discrete and multivariate discrete variants.

\subsection{Hawkes Process}

Let $Y = \{(t_1, m(t_1)),(t_2, m(t_2)),\ldots, (t_n, m(t_n))\}$ denote a data stream where each event time $t_i\in \mathbb{R}$ has an associated mark $m(t_i)>0$ a function of event times. Then the marked Hawkes process $\text{N}$, where $\text{N}_t$ represents the number of events up to time, first introduced by \cite{hawkes1971spectra}, is as a temporal point process with intensity function conditioned on the history of events up to but not including time $t$, denoted by $\text{H}_t = \{Y_s :s < t\}$, as follows: for $t\in \mathbb{R}$,

\begin{align}\label{Hawkes eqn}
    \lambda(t) & =\mathbb{E}\{\text{N}_{t+h}-\text{N}_t|\text{H}_t\}\\
    &= \mu(t) + \sum_{i:t_i<t}{\Phi(t-t_i,m(t_i))},
\end{align} for some small $h>0$. Here $\mu(t)>0$ is the background rate and $t_i$ are the event times prior to time $t$. $m(t_i)$ effects how much the event will excite the process, that is how much the event $t_i$ increases the probability that another event occurs immediately after it (this is known as 'self-excitation') and $v(\cdot, \cdot)>0$ is a continuous excitation function that controls the extent to which events cluster together over time. This process describes a counting process where events increase the probability of further such events in the short term, leading to clustered events. 

\subsection{Discrete Hawkes Process}\label{univariate discete Hawkes}

Consider a marked univariate discrete Hawkes process \cite{Simple} denoted by $\text{N}$, where $\text{N}_t$ represents the number of events up to time interval $(t-1,t]$. $\text{N}_t$ is dependent on the history of events up to but not including time $t$, denoted by $\text{H}_{t-1} = \{Y_s:s \leq t-1\}$, where $Y_s$ represents the observed number of
events in a given time interval $(s-1,s]$. We denote $m(s)$ the associated mark at time interval $(s-1,s]$. Furthermore, $\text{N}_t-\text{N}_{t-1}$ represents the number of event occurrences at time $t\in \mathbb{N}$. Then, we define
	
\begin{align}\label{Discrete Hawkes eqn}
    \lambda(t) & =\mathbb{E}\{\text{N}_t-\text{N}_{t-1}|\text{H}_{t-1}\}\\
    &= \mu(t) + \sum_{i:t_i\leq t-1} {Y_{t_i} \Phi(t-t_i,m(t_i))}.
\end{align} Here $\mu(t)>0$ is the background rate and $t_i$ are the times prior to time $t$ such that an event occurred in the interval $(t_i-1,t_i]$. $m(t_i)$ effects how much the event will excite the process and $v(\cdot, \cdot)>0$ is a discrete excitation function that controls the extent to which events cluster together in time. We assume that $Y_t$ is Poisson distributed, $Y_t \sim \Poisson(\lambda(t))$.

\subsection{Multivariate Discrete Hawkes Process}\label{Multivariate discrete Hawkes}

Consider a marked multivariate discrete Hawkes process $\mathbf{N}$, where $\mathbf{N}_t=\{\text{N}^{(1)}_t,\text{N}^{(2)}_t\ldots,\text{N}^{(M)}_t\}$ represents the number of events up to time interval $(t-1,t]$ in each of the $M$ dimensions of the process. $\mathbf{N}_t$ is dependent on the history of events in every dimension up to but not including time $t$, denoted by $\mathbf{H}_{t-1} = \{\mathbf{Y}_s: s \leq t-1\}$, where $\mathbf{Y}_s=\{Y^{(1)}_s,Y^{(2)}_s, \ldots, Y^{(M)}_s\}$ represents the observed number of
events in a given time interval $(s-1,s]$ in each dimension of the process. We denote $m(s)$ the associated mark at time interval $(s-1,s]$. Furthermore, $\text{N}^{(m)}_t-\text{N}^{(m)}_{t-1}$ represents the number of event occurrences at time $t\in \mathbb{N}$ in dimension $m$. Then, we define

\begin{align}\label{Multivariate discrete Hawkes eqn}
    \lambda^{(m)}(t) & =\mathbb{E}\{\text{N}^{(m)}_t-\text{N}^{(m)}_{t-1}|\mathbf{H}_{t-1}\}\\
     &= \mu^{(m)}(t) + \sum_{l=1}^M \sum_{i:t_{l,i}\leq t-1}{Y^{(l)}_{t_{l,i}} \Phi_{l,m}(t-t_{l,i},m(t_{l,i}))},
\end{align} where $\mu^{(m)}(t)>0$ is the background rate in dimension $m$ and $t_{m,i}$ are the times in dimension $m$ such that prior to time $t$ an event occurred in the interval $(t_{m,i}-1,t_{m,i}]$. The function $m(t_{l,i})$ effects how much the $i^{th}$ event in dimension $l$ will excite the process and $\Phi_{l,m}(\cdot, \cdot)>0$ is a discrete excitation function that controls the extent to which events in dimension $l$ cluster together with events in dimension $m$. We call the increase in the probability of an event in one dimension from an event in another dimension 'cross-excitation'. We assume that, for each $m=1,\ldots,M$, $Y^{(m)}_t$ is Poisson distributed, $Y^{(m)}_t \sim \Poisson(\lambda^{(m)}(t))$.

\section{Model}\label{Model}

We model the occurrences of violent events in the hospital by using a specific form of the multivariate discrete Hawkes as in Section \ref{Multivariate discrete Hawkes}. Our model introduces a particular form of the discrete excitation function, background intensity and formulation of marks. We begin by discussing the choice of the excitation functions.

\subsection{Self- and Cross-Excitation}

Events in different areas of the hospital may occur in contemporaneous bursts, in that instances in a particular ward may be followed by further instances in the same or even another ward in the near future; these bursts can be a result directly witnessing said event, but also a result of hearing about the event via the flow of information via communal areas, which all wards, or particular groups of wards, have access to. This information implies we must not model the excitation in wards individually, but such that events in wards can increase the probability of instances across the hospital, i.e., using a multivariate Hawkes process. Although we have some prior information that wards are grouped into sections of three wards, due to the presence of hospital wide communal areas, it may be unwise to strictly force this structure into the model. We do not include the communal areas in the model since communal areas C1, C2, C3, C4 and C5 from Figure \ref{fig:ward layout} all have less than forty events, too little to fit an accurate model, and communal area C6 represents all other areas of the hospital which patients have access to. This means that two events registered in this section may in not be in a particular area of the hospital. Thus the relationship between an event in C6 and another area is inhomogeneous. Hence, we model the event occurrences using a twelve-dimensional discrete Hawkes process, corresponding to the twelve wards. We have next to discuss the choice of excitation kernels $\Phi_{l,m}(\cdot,\cdot)$ in Equation \eqref{Multivariate discrete Hawkes eqn}. It is common in the continuous-time literature to assume a functional form consisting of an exponential probability density function multiplied by some function of the marks \citep{[5], [11], [15], [22]}. It enforces a sudden increase in probability of event occurrences in the close future, which then exponentially reduces over time. This seems like, in this case, to be a sensible model for the self and cross-excitation behavior. Thus a geometric distribution, being an analog to the continuous-time exponential distribution, may be of interest in this application domain. On top of this, in Section \ref{Estimation}, as in the case of the continuous Hawkes process with the exponential kernel, we are able to write recursive algorithms for the likelihood function and its gradient, that making the estimation more efficient. We use the following form for the excitation kernel

\begin{equation}\label{excitation kernel form}
    \Phi_{l,m}(t-t_{l,i},m(t_{l,i}))=f_{l,m}(m(t_{l,i}))g_{l,m}(t-t_{l,i}), \quad \forall l,m=\{1,2,...,12\},
\end{equation} where $g_{l,m}(t)=\beta(1-\beta)^{t-1}$, the geometric distribution, which determines the shape of the decay function over time from ward $l$ to ward $m$ and $f_{lm}(\cdot)$ controls the extent to which events cluster together across wards $l$ to $m$ and is a function of the marks. Again, for the sake of efficient estimation we fix $\beta \in (0,1)$ to be the fixed over the wards. Letting this parameter be variable between the wards introduces an extra 143 parameters to the model, thus introducing identifiablity issues.

\subsection{Covariate Data}

In addition to the excitation exhibited in the same and between wards due to either directly observing the event or hearing about the event via word of mouth, there is an additional way information flows through the hospital, namely, via an alarm. The reason for an alarm being sounded is due to a particularly violent event, as determined by nearby on-shift employees, which requires urgent assistance from employees in other areas of the hospital. These alarms can be heard throughout the hospital, which implies that we need to model the added excitation due an event which triggered an alarm differently from a standard event. Present in the data is an indicator for whether the alarm was sounded in response to each event. We implement this data into our model via a mark. As such, we define the function $f_{l,m}(\cdot)$ as follows, 

\begin{equation}\label{K form}
    f_{l,m}(m(t_{l,i}))=K_{l,m}+\alpha_{l,m} \mathbbm{1}(m(t_{l,i})), \quad \forall l,m=\{1,2,...,12\},
\end{equation} where we define $m(t_{l,i})=1$ if there was an alarm sounded in the interval $(t_{l,i}-1,t_{l,i}]$ caused by an event in ward $l$ and $m(t_{l,i})=0$ otherwise. We define the indicator function $\mathbbm{1}(x)=1$ if $x=1$ and $\mathbbm{1}(x)=0$ otherwise. In this setup, $K_{lm}$ controls the extent to which events cluster together caused by an event which did not sound an alarm, and $K_{lm}+\alpha_{l,m}$ controls the extent to which events cluster together caused by an event which caused an event to be sounded. We can thus interpret $\alpha_{l,m}$ as the extra excitation caused by an alarm.

\subsection{Seasonality}\label{Model: Seasonality}

From our preliminary analysis in Section \ref{Data}, we found that there is evidence of time-of-day seasonality in the data. To include this seasonal component in our model, we allow the background rate to vary deterministically over each day. In particular, we define a function $s:\{1,2,...,24\} \to (0,1)^{24}$ that gives the proportion of events that have occurred in that hour averaged over the training period and over each ward. Then define the background rate in dimension $m$ to be 
\begin{equation}\label{background in ward m definition}
    \mu^{(m)}(t) = \mu^{(m)} \cdot s(h(t)),
\end{equation}
where $h(t)$ gives the hour of day of the $t^{th}$ interval and $\mu^{(m)}>0$ is constant to be estimated, i.e., the ward-specific constant background scaling constants. We treat $s(h(t))$ as fixed then estimate $\mu^{(m)}$, $m=1,\ldots,M$ within the maximum likelihood estimation. We make the assumption that the time-of-day seasonality between wards are the same. This particular parameterization means that we can include time-varying background rates while only using twelve parameters. Having a small number of parameters helps to avoid overfitting. We tried several alternative approaches to modelling seasonality of the background rate, but this approach had the best predictive performance.

\subsection{Regularization}\label{regularization}

We now summarize our model up until here. To do so we introduce a piece of notation for simplicity of presentation. Namely, denote the set of arrival times in dimension $m$ such that the events caused an alarm to be sounded by $A = \{ A^{(1)},\ldots,A^{(M)} \}$, where $A^{(m)}=\{a_{m,1},\ldots,a_{m,N_{a,m}} \}$ for each $m =1,\ldots,M$, and denote $N_{a,m}$ the number of events in dimension $m$ that caused an alarm to be sounded. Then the conditional intensity function of our model is defined by  

\begin{equation}\label{Alarm lambda}
    \lambda^{(m)}(t) = \mu^{(m)} \cdot s(h(t)) + \sum^M_{l=1} \sum_{i:t_{l,i}\leq t-1} Y^{(l)}_{t_{l,i}} \Big( K_{l,m} + \alpha_{l,m} \mathbbm{1}(t_{l,i} \in A^{(l)}) \Big)  \beta(1-\beta)^{t-t_{l,i}-1}, 
\end{equation} where $h(t)$ gives the hour of the day corresponding to the $t^{th}$ interval, $s(\cdot)$ gives the proportion of events that have occurred in that hour averaged over each ward, $\beta \in (0,1)$, $K_{l,m}>0$ and $\alpha_{l,m}>0$ $\forall l,m\in \{q,\ldots,M\}$. 

Denote the collection of parameters in our model by $\theta= \{\mu^{(1)},\ldots,\mu^{(M)}, K_{1,1},\ldots,K_{M,M}, \beta,  \alpha_{1,1},\ldots,\alpha_{M,M} \}$. That is, we have 301 parameters to train. To avoid over-fitting, it is important to regularize our model. To do so, we use a (ridge) penalty term in the likelihood function. Namely, we define our regularized log-likelihood function $\log L_{R}(\cdot)$ by 

\begin{equation}\label{regularized log-likelihood}
    \log L_{R}(\theta) = \log L(\theta) + \lambda_h P(K,\alpha),
\end{equation}

where $\log L(\cdot)$ gives the non-regularized log-likelihood function of the model described by Equation \eqref{Alarm lambda} and $P(k,\alpha)=\sum_{l \neq m}K_{l,m}^2+ \sum_{l,m=1}^M\alpha_{l,m}^2$. That is, the penalisation term is given by the sum of the squared off-diagonal terms in the $K$-matrix multiplied by the hyperparameter $\lambda_h>0$. The motivation behind this choice of penalization is to embed our prior belief that the act of witnessing a violent event first hand should have a greater effect than learning about the event via word of mouth. This regularization method shrinks the off-diagonal elements $K$ and the matrix $\alpha$.  Note that we do not enforce any particular structure in the matrix $\alpha$, that is the extra excitation caused by the alarm, rather we regularize every part equally. We find the optimal $\lambda_h$ by cross-validating against a separate validation set $\textbf{T}_{val}$.

\section{Estimation}\label{Estimation}

In this section we work with the model as defined by Equation \eqref{Alarm lambda}. Under the multivariate discrete marked Hawkes process, $Y^{(m)}_t$ is Poisson distributed, as described in Section \ref{Multivariate discrete Hawkes}, That is, 

\begin{equation}\label{Probability under discrete Hawkes}
    P\Big(Y^{(m)}_t = y \Big| \lambda^{(m)}(t)\Big)=\frac{\Big(\lambda^{(m)}(t)\Big)^y e^{-\lambda^{(m)}(t)}}{y!}.
\end{equation} Let the time-frame of observations be denoted by $[0,T]$ and that there are $\delta$ discrete time-stamps over each time period. From the above, we obtain the following log-likelihood function, 

\begin{equation}\label{log-likelihood 1}
    \log L(\theta|\tau) = \sum^{T\delta}_{t=1} \sum^M_{m=1}\Big( Y^{(m)}_t \log(\lambda^{(m)}(t))-\lambda^{(m)}(t) \Big),
\end{equation} where $\theta$ is the set of parameters in the intensity function to be estimated and $\tau = \{ \tau^{(1)},\ldots,\tau^{(M)} \}$ where $\tau^{(m)}=\{t_{m,1},\ldots,t_{m,N_m} \}$ is the set of arrival times in dimension $m$. Hence $\tau$ is the set of arrivals times over the entire process. The problem with directly using this form of the likelihood is that if we observe the data over a large time-frame, $T>>0$, and/or a over large number of discrete steps, $\delta>>0$, then the sum becomes very expensive. Summing the intensity function over $T\delta$ is the problem since these terms have complexity $\mathcal{O}((T\delta)^2)$. 
However, since the excitation function in Equation \eqref{Alarm lambda} is decomposed in the form as in Equation \eqref{excitation kernel form}, i.e.,

\begin{equation*}\label{multivariate intensity geom - estimation first}
    \Phi_{l,m}(t-t_i,m(t_{l,i}))=f_{l,m}(m(t_{l,i}))g_{l,m}(t-t_{l,i}),
\end{equation*} for each $l,m=\{1,2,...,M\}$, we are able to reduce the computational complexity of the log-likelihood function in Equation \eqref{log-likelihood 1} to order $\mathcal{O}(N_{\tau}^2)$, where $N_{\tau}$ gives the total number of events over the entire process. The proof of this following proposition is analogous to the continuous-time setting \cite{Ozaki}. Let $N_{\delta}=T\delta$.

\begin{prop}\label{First likelihood prop}
For the intensity function as defined by Equation \eqref{multivariate intensity geom - estimation first}, we can rewrite the log-likelihood function in Equation \eqref{log-likelihood 1} as follows

\begin{equation}\label{log-likelihood 2}
    \log L(\theta|\tau) = \sum^M_{m=1} \sum_{t\in \tau^{(m)}}  Y^{(m)}_t\log(\lambda^{(m)}(t))
    - \sum^M_{m=1}\sum^M_{l=1} \sum_{t\in \tau^{(l)}} Y_t^{(l)} f_{l,m}(m(t)) G_{l,m}(N-t) 
    - \sum_{m=1}^M \sum_{t=1}^{N_{\delta}} \mu^{(m)}(t)
\end{equation}

where $G_{l,m}(\cdot)$ is the cumulative mass function of the probability density function $g_{l,m}(\cdot)$. 

\end{prop}

\begin{proof}
    See Appendix \ref{First likelihood prop - proof}.
\end{proof}

We can apply Proposition \ref{First likelihood prop} using the form of the intensity in Equation \eqref{Alarm lambda} to obtain the following corollary.

\begin{corollary}\label{corollary}
For the intensity function as defined by Equation \eqref{Alarm lambda}, we can rewrite the log-likelihood function in Equation \eqref{log-likelihood 1} as follows

\begin{multline}
    \log L(\theta|\tau, A) = \sum^M_{m=1}\sum_{t\in \tau^{(m)}} Y_t^{(m)} \log(\lambda^{(m)}(t)) - \sum_{m=1}^M \sum_{l=1}^M K_{l,m} \sum_{t\in \tau^{(l)}} Y_t^{(l)} (1-(1-\beta)^{N_{\delta}-t}) \\
    -\sum_{m=1}^M \sum_{l=1}^M \alpha_{l,m} \sum_{a\in A^{(l)}} Y_a^{(l)} (1-(1-\beta)^{N_{\delta}-a}) - \sum^M_{m=1} \mu^{(m)} \sum^{N_{\delta}}_{t=1} s(h(t)).
\end{multline}

\end{corollary}

\begin{proof}
    See Appendix \ref{corollary - proof}.
\end{proof}

We have achieved a reduction of the complexity of the intensity function calculation by the fact that $N_{\tau} \leq N_{\delta}$. In fact, in our application, $N_{\tau}$ is much smaller than $N_{\delta}$. Thus the reduction in complexity is quite dramatic. Still, however, when the number of events is high, the repeated intensity function calculation will become expensive. We develop a faster approach in the next section which takes only linear  time complexity to calculate the intensity function. This is achieved by developing a recursive formulation for the intensity function by maintaining two additional data structures which take a constant space per event.

\begin{prop}\label{Recursion 1} For $j\in \{1,\ldots,N_m\}$ and $m \in \{1,\ldots,M\}$, we can write the conditional intensity function in Equation \eqref{Alarm lambda} in terms of recursive relations $R$ and $R_a$ as follows, 

\begin{equation}
    \lambda^{(m)}(t_{m,j})= \mu^{(m)}\cdot s(h(t_{m,j})) + \sum^M_{l=1}{K_{l,m}R(j,l)} + \sum^M_{l=1}{\alpha_{l,m}R_a(j,l)}
\end{equation}

where $R$ and $R_a$ are defined as follows: $\forall l \in \{1,...,M\}, j=1,...,N_m-1$, define

\begin{equation}\label{R}
    R(j+1,l) := (1-\beta)^{t_{l,j+1}-t_{l,j}}R(j,l)+Y_{t_{j,l}}^{(l)}\beta(1-\beta)^{t_{l,j+1}-t_{l,j}-1},
\end{equation}

and 

\begin{equation}\label{R_a}
    R_a(j+1,l) := (1-\beta)^{t_{l,j+1}-t_{l,j}}R_a(j,l)+Y_{t_{l,j}}^{(l)} \mathbbm{1}(t_{l,j} \in A^{(l)})\beta(1-\beta)^{t_{l,j+1}-t_{l,j}-1},
\end{equation}

where $R(1,l):=R_a(1,l):=0$.

\end{prop}

\begin{proof}
    See Appendix \ref{Recursion 1 - proof}.
\end{proof}

Proposition \ref{Recursion 1} allows us to write a recursive algorithm to compute the intensity function over each event time. At every step of the sum over $\tau^{(l)}$, we have that the complexity of $R(j,l)$ given the value $R(j-1,l)$ is $\mathcal{O}(1)$. In contrast, without the recursive relation, the same sum would have complexity $\mathcal{O}(N_l^2)$.

We approach the estimation of our process in a frequentist manner. That is, we use a gradient ascent algorithm to traverse the parameter space to obtain the maximum likelihood estimates. Having a closed form gradient will reduce the issues with numerical inaccuracy compared to a finite difference approximation of the algorithm. However, the closed form solution of the gradient shares the same burden of the likelihood function. That is, the computational complexity to compute an instance of the gradient is $\mathcal{O}(N_{\tau}^2)$, as shown in Appendix \ref{partials - alarm}. However, in a similar manner as with the log-likelihood function as above, the following proposition shows that we can write the gradient in a recursive manner, which allows us to reduce the complexity of this calculation to $\mathcal{O}(N_{\tau})$.

\begin{prop}\label{corr 1} 

We can write the gradient of the log-likelihood function as in Corollary \ref{corollary}, in terms of recursive relations as follows; for $j\in \{1,\ldots,N_l\}$ and $l \in \{1,\ldots,M\}$,

\begin{equation}\label{mu gradient - recursive}
    \frac{\partial}{\partial \mu^{(p)}} \log L(\theta|\tau,A) = \sum_{j:t_{p,j} \in \tau^{(p)}} \frac{Y_{t_{p,j}}^{(p)}s(h(t_{p,j}))}{\mu_0^{(p)}\cdot s(h(t_{p,j})) + \sum_{l=1}^M k_{l,p} R(j,l) + \sum_{l=1}^M \alpha_{l,p} R_a(j,l)} - \sum^{N_{\delta}}_{t=1} s(h(t))
\end{equation}

\begin{multline}\label{k gradient - recursive}
    \frac{\partial}{\partial K_{pq}} \log L(\theta|\tau,A) = \sum_{j:t_{q,j} \in \tau^{(q)}} \frac{Y_{t_{q,j}}^{(q)} R(j,p)}{\mu_0^{(q)}\cdot s(h(t_{q,j})) + \sum_{l=1}^M k_{l,q} R(j,l) + \sum_{l=1}^M \alpha_{l,q} R_a(j,l)} \\
    -
    \sum_{t\in \tau^{(p)}}Y_t^{(p)}(1-(1-\beta)^{N_{\delta}-t})
\end{multline}

\begin{multline}\label{alpha gradient - recursive}
    \frac{\partial}{\partial \alpha_{p,q}} \log L(\theta|\tau,A) = \sum_{j:t_{q,j} \in \tau^{(q)}} \frac{Y_{t_{q,j}}^{(q)} R_a(j,p)}{\mu_0^{(q)}\cdot s(h(t_{q,j})) + \sum_{l=1}^M k_{l,q} R(j,l) + \sum_{l=1}^M \alpha_{l,q} R_a(j,l)} \\
    -
    \sum_{a\in A^{(p)}}Y_a^{(p)}(1-(1-\beta)^{N_{\delta}-a})
\end{multline}

\begin{multline}\label{beta gradient - Recursive}
    \frac{\partial}{\partial \beta} \log L(\theta|\tau,A)
    =  \sum_{m=1}^{M} \sum_{j:t_{m,j} \in \tau^{(m)}} \frac{Y_{t_{m,j}}^{(m)}}{\mu^{(m)}\cdot s(h(t_{m,j}))  + \sum_{l=1}^M K_{l,m} R(j,l) + \sum_{l=1}^M \alpha_{l,m} R_a(j,l)} \times \\ 
    \Bigg( \sum_{l=1}^{M} K_{l,m} \Big( R^{1}(j,l)-R^{2}(j,l)\Big) + 
    \sum_{l=1}^{M} \alpha_{l,m} \Big( R_a^{1}(j,l)-R_a^{2}(j,l)\Big)\Bigg)  \\
    - \sum_{m=1}^{M} \sum_{l=1}^{M}  \sum_{t\in \tau^{(l)}} K_{l,m} Y^{(l)}_t (N_{\delta}-t)(1-\beta)^{N_{\delta}-t-1} 
    -  \sum_{m=1}^{M} \sum_{l=1}^{M} \sum_{a\in A^{(l)}} \alpha_{l,m} Y^{(l)}_a (N_{\delta}-a)(1-\beta)^{N_{\delta}-a-1}.
\end{multline} $R$ and $R_a$ are defined as in Proposition \ref{Recursion 1} and $R^1$, $R^1_a$, $R^2$ and $R^2_a$ are as follows: Define $R^{1}(1,l):=R_a^{1}(1,l):=R^{2}(1,l):=R_a^{2}(1,l):=0$ $\forall l = 1,...,M$ and we define recursively, for $j\in \{1,\ldots,N_l-1\}$ and $l \in \{1,\ldots,M\}$,

\begin{align}
    R^{1}(j+1,l) &:= (1-\beta)^{t_{l,j+1}-t_{l,j}}R^{1}(j,l)+Y_{t_{l,j}}^{(l)}(1-\beta)^{t_{l,j+1}-t_{l,j}-2} \\ 
    R_a^{1}(j+1,l) &:= (1-\beta)^{t_{l,j+1}-t_{l,j}}R_a^{1}(j,l)+Y_{t_{l,j}}^{(l)}\mathbbm{1}(t_{l,j}\in A^{(l)})(1-\beta)^{t_{l,j+1}-t_{l,j}-2} \\
    R^{2}(j+1,l) &:= \beta(t_{l,j+1}-t_{l,j})R^{1}(j+1,l)+(1-\beta)^{t_{l,j+1}-t_{l,j}}R^2(j,l) \\
    R_a^{2}(j+1,l) &:= \beta(t_{l,j+1}-t_{l,j})R_a^{1}(j+1,l)+(1-\beta)^{t_{l,j+1}-t_{l,j}}R_a^2(j,l)
\end{align}
\end{prop}

\begin{proof}
    See Appendix \ref{corr 1 - proof}.
\end{proof}

\section{Forecasting}

In the frequentist setting, forecasting under a Hawkes process consists of conditioning the intensity function on the known past events and then simulating many paths of the process, giving a predictive distribution for the number of events in the given prediction window. Simulation algorithms for continuous-time Hawkes process rely on a technique known as thinning \citep{article}. Despite the derivation being less complicated compared to the simulation method for the continuous-time variant, the methods for discrete Hawkes processes \citep{Poisson_Simulation} are burdened with having to run computations over each discrete time-step. We will tackle this issue in this section in order to develop algorithms for the purpose in forecasting, the results for which are presented in the next section. We present a naive simulation algorithm \ref{alg:Naive - Alarms} and a novel simulation algorithm \ref{alg:recursive Alarm}, for the model as described in Section \ref{Model}, which reduces the complexity of Algorithm \ref{alg:Naive - Alarms} by reducing the complexity of the intensity function calculation which has to be computed at each step. 

To simulate from our model defined by the intensity function as in Equation \eqref{Alarm lambda}, we must simulate alarm indicators as well as event times. Hence, we need a model for the alarms. We model the alarm indicators, given at least one event has occurred, by a Bernoulli distribution. That is we assume that there is a constant probability that each event will cause an alarm to be sounded, denoted by $p:=p(i \in A^{(m)})$, for each $m=1,\ldots,M$. We can estimate $p$ by calculating the maximum likelihood estimator for the parameter, i.e., the sample mean. For our application the alarms are used as a covariate where we do not explicitly assume a distribution for these occurrences. So in practice, one could simulate the alarm process using any distribution.

\begin{algorithm}[H]
	\caption{Naive algorithm to simulate from the multivariate discrete Hawkes model with intensity function as in Equation \eqref{Alarm lambda}.}
	\label{alg:Naive - Alarms}
	\begin{algorithmic}[1]
	\State \textbf{Require} $N_{sim}, p,\mu^{(1)},\ldots,\mu^{(M)}, K_{1,1},\ldots,K_{M,M}, \beta,  \alpha_{1,1},\ldots,\alpha_{M,M}$
        \For {$m=1,\ldots,M$}
        \State Generate $u \sim \Uniform(0,1)$
        \If{$u<p$} 
				\State Set $1\in A^{(m)}$.
		\EndIf
        \EndFor
	\For {$m=1,\ldots,M$}
                \State Generate $Y^{(m)}_1 \sim \Poisson(\mu^{(m)} \cdot s(h(1)))$
        \EndFor
		\For {$t=2,\ldots,N_{sim}$}
			\For {$l=1,2,\ldots,M$}
				\State $\lambda^{(m)}(t) \gets \mu^{(m)} \cdot s(h(t)) + \sum^M_{l=1} \sum_{i:t_{l,i}\leq t-1} Y^{(l)}_{t_{l,i}} \Big( K_{l,m} + \alpha_{l,m} \mathbbm{1}(m(t_{l,i})) \Big)  \beta(1-\beta)^{t-t_{l,i}-1} $
                \State Generate $Y^{(m)}_t \sim \Poisson(\lambda^{(m)}(t))$
                \State Generate $u \sim \Uniform(0,1)$
                \If{$Y^{(m)}_t>0$ and $u<p$} 
				\State Set $t\in A^{(m)}$
				\EndIf
			\EndFor
		\EndFor
	\end{algorithmic} 
\end{algorithm}

We show next that the intensity function can be written recursively which will in-turn allow us to design a simulation algorithm which has computational complexity $\mathcal{O}(M \cdot N_{\delta})$, compared to $\mathcal{O}(M \cdot (N_{\delta})^2)$ for the naive Algorithm \ref{alg:Naive - Alarms}.

\begin{lem}\label{R3 - with alarm}
We can write the intensity function as defined in Equation \eqref{Alarm lambda}, for each discrete time step $t=1,\ldots,N_{\delta}$ as follows 

\begin{equation}
    \lambda^{(m)}_t = \mu^{(m)}(t)+\sum^M_{l=1}K_{lm}R^3(t,l)+\sum^M_{l=1}\alpha_{l,m}R_a^3(t,l), 
\end{equation}

where $\forall l=\{1,...,M\}$ we define $R^3(1,l):=R_a^3(1,l):=0$, 

\begin{equation}
    R^3(t+1,l)=(1-\beta)R^3(t,l)+Y^{(l)}_t\beta, \; \quad \text{for} \hspace{0.1cm} t=1,\ldots,N_{\delta}-1,
\end{equation}

and 

\begin{equation}
    R_a^3(t+1,l)=(1-\beta)R_a^3(t,l)+Y^{(l)}_t\mathbbm{1}(t\in A)\beta, \; \quad \text{for} \hspace{0.1cm} t=1,\ldots,N_{\delta}-1.
\end{equation}

\end{lem}

\begin{proof}
    See Appendix \ref{R3 - with alarm - proof}.
\end{proof}

Using Lemma \ref{R3 - with alarm}, we can now present a recursive algorithm to simulate from the multivariate discrete Hawkes model with intensity function as in Equation \eqref{Alarm lambda}. This algorithm is such that at each time step $t$, we must only compute $R^3(t,l)$, for $l = 1,\ldots,M$, conditional on the values $R^3(t-1,l)$, $l= 1,\ldots,M$. Hence, at each time-step, the computational complexity is of order $\mathcal{O}(M \cdot N_{\delta})$. In comparison, the Naive Algorithm \ref{alg:Naive - Alarms} has complexity $\mathcal{O}(M \cdot (N_{\delta})^2)$ due to the computation of $\lambda^{(m)}(t)$, for each $t=1,\ldots,N_{\delta}$.

\begin{algorithm}[H]
	\caption{A recursive algorithm to simulate from the multivariate discrete Hawkes model with intensity function as in Equation \eqref{Alarm lambda}.}
        \label{alg:recursive Alarm}
	\begin{algorithmic}[1]
	\State \textbf{Require} $N_{sim}, p,\mu^{(1)},\ldots,\mu^{(M)}, K_{1,1},\ldots,K_{M,M}, \beta,  \alpha_{1,1},\ldots,\alpha_{M,M}$
	\For {$m=1,\ldots,M$}
                \State Generate $Y^{(m)}_1 \sim \Poisson(\mu^{(m)} \cdot s(h(1)))$
        \EndFor
        \For {$m=1,\ldots,M$}
        \State Generate $u \sim \Uniform(0,1)$
        \If{$u<p$} 
				\State Set $1\in A^{(m)}$.
		\EndIf
        \EndFor
	\State Set $R^3(1,l) \gets 0$ and $R_a^3(1,l) \gets 0$.
		\For {$t=2,\ldots,N_{sim}$}
		\For {$m=1\ldots,M$}
			\For {$l=1,2,\ldots,M$}
			    \State $R^3(t,l)=(1-\beta)R^3(t-1,l)+Y^{(l)}_{t-1}\beta$
			    \State $R_a^3(t,l)=(1-\beta)R^3(t-1,l)+Y^{(l)}_{t-1}\mathbbm{1}(t-1\in A^{(l)})\beta$
			    \EndFor
				\State $\lambda^{(m)}(t) \gets \mu^{(m)} \cdot s(h(t)) + \sum^M_{l=1}{K_{l,m}R^3(t,l)} + \sum^M_{l=1}{\alpha_{l,m}R_a^3(t,l)}$
				\State Generate $Y^{(m)}_t \sim \Poisson(\lambda^{(m)}(t))$
				\State Generate $u \sim \Uniform(0,1)$
				\If{$Y^{(m)}_t>0$ and $u<p$} 
				\State Set $t\in A^{(m)}$
				\EndIf
		\EndFor
	\EndFor
	\end{algorithmic} 
\end{algorithm}

\section{Results}\label{Results}

We fit four model variations to the twelve wards in the forensic psychiatric hospital between the dates 26$^{th}$ June 2012 to 1$^{st}$ January 2019. We denote time interval over which we train our models as $T_{train}$. We use a validation, $10\%$ of the dataset, to perform cross-validations on the model for the regularisation techniques described below. We denote this set by $\textbf{T}_{val}$. A hold out test set over the final $15\%$ of the dataset is used to evaluate the predictive performance of the model variations. We denote this test interval as $\textbf{T}_{test}$. We first define four models which we will be comparing their in-sample and predictive performance on our hospital data. Our first and simplest model is the homogeneous Poisson process. We then define our next comparative model to be twelve individual univariate Hawkes processes relating to each ward. Next we define a multivariate Hawkes process and finally a multivariate Hawkes process with alarm covariates as defined in Section \ref{Model}. More formally, we define our models for comparison as follows

\begin{itemize}
  \item \textbf{Inhomogeneous  Poisson Process (IPP).} Consider an inhomogeneous Poisson process where the number of events in the interval $t$ in the $m^{th}$ ward is Poisson with rate
  
  $$\lambda^{(m)}(t) = \mu^{(m)} \cdot s(h(t)),$$ 
  
  where $s(h(t))$ is the seasonal component defined in \ref{Model: Seasonality}. \newline 
  
  \item \textbf{Univariate Hawkes Processes (UHP).} Let each ward in the forensic hospital be a univariate discrete Hawkes processes as defined in Section \ref{univariate discete Hawkes}. That is, the intensity function is defined as

\begin{equation*}
    \lambda^{(m)}(t) = \mu^{(m)} \cdot s(h(t)) + K_m\sum_{i:t_{m,i} \leq t-1}{Y^{(m)}_{t_{m,i}} \beta(1-\beta)^{t-t_{m,i}-1}}.
\end{equation*}

where $K_m>0$, for each $m=1,\ldots,M$, and $s(h(t))$ is the seasonal component defined in Section \ref{Model: Seasonality}. \newline
  
  \item \textbf{Multivariate Hawkes Process (MHP/MHP-R).} Consider an twelve-dimensional discrete Hawkes process as defined in Section \ref{Multivariate discrete Hawkes}, where we model each ward by a dimension in the process. That is, 

\begin{equation*}
    \lambda^{(m)}(t) =  \mu^{(m)} \cdot s(h(t)) + \sum^M_{l=1}{K_{lm}\sum_{i:t_{l,i}\leq t-1}{Y^{(l)}_{t_{l,i}} \beta(1-\beta)^{t-t_{l,i}-1}}}.
\end{equation*}

where $s(h(t))$ is the seasonal component defined in Section \ref{Model: Seasonality}. We implement a Ridge penalty, with hyperparameter $\lambda_h^{MHP}$, in the likelihood function where we implement a ridge penalty on the off-diagonal elements in the $K$-matrix and denote this model by by \textbf{MHP-R}.
  
  \item \textbf{Multivariate Hawkes Process with Alarm Covariates (MHPA/MHPA-R).} Consider a twelve-dimensional discrete Hawkes process with the sounding of alarms used as marks in the process, as in Section \ref{Model}, where we model each ward by a dimension in the process. The intensity function is as follows,

\begin{equation*}
    \lambda^{(m)}(t) = \mu^{(m)} \cdot s(h(t)) + \sum^M_{l=1} \sum_{i:t_{l,i}\leq t-1} Y^{(l)}_{t_{l,i}} \Big( K_{l,m} + \alpha_{l,m} \mathbbm{1}(m(t_{l,i})) \Big)  \beta(1-\beta)^{t-t_{l,i}-1}.
\end{equation*}

We define a regularized model with the optimal hyperparameter, as determined against the validation set by \textbf{MHPA-R}. That is, we implement a Ridge penalty, with hyperparameter $\lambda_h^{MHPA}$, in the likelihood function where we penalise the off-diagonal elements in the $K$-matrix and all elements of the $\alpha$-matrix, as described in Section \ref{regularization}.

\end{itemize}

\subsection{Predictive Model Comparison}
We use the predictive log-likelihood \cite{Daniel_forecasting} on the events in the test data, given the previous observed event counts, as a metric for predictive performance. More formally, suppose we have observations at time points $1,\ldots,n$ denoted by $\textbf{Y}=(Y_1,\ldots,Y_n)^{'}$, and we would like to calculate the probability of observing the out-of-sample observation $Y_{n+1}$. To do so, we condition on all previous observations and our model parameters fitted on the in-sample data. That is we define the predictive log-likelihood for the observation $Y_{n+1}$ to be, 

\begin{equation}\label{pLL equation - 1 obs}
    pLL(y_{n+1}|\textbf{Y},\hat{\theta}) = p(Y_{n+1}|\textbf{Y},\hat{\theta}),
\end{equation}

where $\hat{\theta}$ is the MLE estimated on a subset of the in-sample data (in our application we use part of the in-sample data as a validation set for finding the optimal hyperparameters). To calculate the probability of observing $\textbf{Y}_{test}=(Y_{n+1},\ldots,Y_{n+m})^{'}$, we calculate the probability of observing each event in the test set conditioned on all previous observations, including the earlier events in the test set. That is, 

\begin{equation}\label{pLL equation}
    pLL(\textbf{Y}_{test}|\textbf{Y},\hat{\theta}) = \prod_{k=1}^{m} p(Y_{n+k}|\textbf{Y}, Y_{n+1},...,Y_{n+k-1},\hat{\theta}).
\end{equation}

We calculate the predictive log-likelihood for each of our models defined at the beginning of this section, using the notation $pLL$ superscript by the model abbreviations. We use the notation '-R' after the model abbreviation to denote the optimal regularized model. When commenting on the results, for the sake of presentation, we only refer to the model abbreviation. That is, if there is increased predictive performance due to regularization, we will be referring to the regularized model. (We found that the optimal hyperparameters are $\lambda_h^{MHP} = 0.5$ and $\lambda_h^{MHPA} = 0.5$). We report these for each ward individually and for the entire hospital. The best (i.e. highest predictive log-likelihood) model for each product is shown in bold. We use this as a metric for each models predictive performance. Note that ward numbers in Table \ref{tab:pLL} relate to those of the ward layout described in Figure \ref{fig:ward layout}. 

Table \ref{tab:pLL} reveals some interesting findings. First, all models beat the simplest model, IPP, which forecasts based on a fixed background rate for each ward with a daily-seasonal component. We observe that the UHP model provides quite a significant improvement in predictive performance compared to the IPP model. That is, the model benefits from the inclusion of self-excitation in the model. Next, we extend the univariate model variation UHP to its multivariate variant MHP. This improvement in predictive performance indicates that the count process not only benefits from the inclusion of self-excitation in the model but also from cross-excitation between wards. That is, there is some information flow between wards which increases the rate of event occurrences. We see a slight increase in predictive performance going from the MHP model to the regularized MHP model. Finally, our proposed model MHPA with and without regularization outperforms all other models overall. This implies that the presence of the alarm as a covariate plays a significant role in the model, but due to the large number of parameters being introduced when compared to the MHP model (144 extra parameters), regularization is needed in order to reduce the chance of overfitting. 

We can illustrate how the models differ by reporting the average instance triggering probabilities due to the background rate, self/cross-excitation and the extra alarm self/cross-excitation over each event. Table \ref{tab: triggering probabilities} and Figure \ref{fig:proportion graphs} presents some findings that distinguish the models introduced at the beginning of this section. This gives some indication to the reasons behind the increased predictive performance of each model. As we can see, on average, for the UHP model, the probability that each event was due to excitation was about 13\%. Whereas the IPP model, by construction, is dependant only on the background rate. That is a large proportion of the events occurred due to excitation. Going from the UHP model to the multivariate version MHP, Table \ref{tab: triggering probabilities} shows that on average the probability of any event being due to excitation is about 2\% larger. In fact, we can see that in fact 79\% of the excitation was due to cross excitation between the wards. The benefit of this is indicated by the improved predictive performance in the latter model. We do not see any differences (in reported significant figures) between MHP and its regularized counterpart MHP-R, which is indicative of the similar predictive performance of the two models in Figure \ref{tab:pLL}. By including the parameters representing the effect of the alarms with the MHPA model and comparing with the MHP and MHP-R models, we see that on average the probability that an event is caused by the extra excitation due to an alarm being sounded was nearly 2\% on average. 
We can see that this excitation comes in turn for reduced probability that the events were caused by non-alarm excitation and the effect of the background rate stays constant. This could imply that part of the excitation captured in the MHP model was due to the alarms being sounded and this added information could in turn be a reason for the increased predictive performance between the models. When we look at the regularized model MHPA-R, we can see that there is not any significant change in the alarm excitation, still sitting at 2\%, but rather an increase in background excitation and a decrease in non-alarm excitation. A further observation to make here is the split between self-excitation (excitation in the same ward) and cross-excitation (excitation between wards). We can see that 75\% of the alarm excitation is cross-excitation. Implying that the effect of an alarm is more spread out throughout the hospital. In contrast, 80\% of the the non-alarm excitation is in the same ward (this is similar in the other multivariate models). Implying that the Alarm is three times more of a trigger in other areas of the hospital than it is in the same ward. 

The average extra alarm excitation of $2\%$ in the MHPA-R model may seem quite low at first glance. However, since these probabilities peak after an alarm has been sounded, and since alarms are exclusively for rare events, it is beneficial to visualise the peaks of the probabilities over a few events. Figure \ref{fig:proportion graphs} shows, from left-to-right over a period in the middle of our dataset, gives the probabilities that each event was caused by the background rate, by events that had not sounded an alarm and the probability that the events were caused by the extra excitation caused by an alarm. For example we can see there are large spikes around the $45^{th}$ event. In fact, in this region the probability that the event was caused by the extra alarm excitation rose to about $75\%$. Hence, it is more likely that the event was caused by the extra excitation caused by a sounded alarm than any other form of excitation or the background rate.

\begin{table}
\centering
\begin{tabular}{c c c c c c c} 
\hline
Ward & $pLL^{IPP}$ & $pLL^{UHP}$ & $pLL^{MHP}$  & $pLL^{MHP-R}$ & $pLL^{MHPA}$ & $pLL^{MHPA-R}$ \\
\hline
   1 & -523.6 & -496.8 & \textbf{-497.3} &  \textbf{-497.3} & \textbf{-497.3} &  \textbf{-497.3}\\
    2 & -51.2 & \textbf{-49.7} & -49.8 &  -49.8 & -50.0 &  \textbf{-49.7}\\
    3 & -203.8 & -194.6 & -194.6 &  -194.6 & -194.6 & \textbf{-194.4}\\
    4 & -364.0 & -348.7 & \textbf{-348.6} & \textbf{-348.6} & -348.7 & -348.9 \\
    5 & -1551.1 & -1520.5 & \textbf{-1518.1} & \textbf{-1518.1} & -1519.6 & -1518.8\\
    6 & -234.9 & -222.2 & -222.0 & -222.0 & \textbf{-220.9} & -221.0\\
    7 & -313.3 & -294.7 & -294.7 & -294.7 & \textbf{-292.9} & \textbf{-292.9}\\
    8 & -64.1 & \textbf{-61.3} & -61.6 & -61.6 & -61.6 & -61.6\\
    9 & -82.6 & -79.0 & \textbf{-79.4} & \textbf{-79.4} & \textbf{-79.4} & -79.5\\
    10 & -1029.4 & -968.4 & -968.4 & -968.4 & -967.9 & \textbf{-967.6}\\
    11 & -420.4 & -404.3 & -404.5 & -404.5 & -403.9 & \textbf{-403.6}\\
    12 & -108.5 & -101.1 & -101.1 & -101.1 & \textbf{-101.0} & \textbf{-101.0}\\
    \hline
    Overall & -4946.9 & -4742.0 & -4740.1 & -4740.1 & -4737.8 & \textbf{-4736.4} \\
\hline
\end{tabular}
\caption{Log-likelihood values under each of the specified models on the test data set, $\textbf{T}_{test}$.}
\label{tab:pLL}
\end{table}

\begin{table}
\centering
\begin{tabular}{c c c c} 
\hline
Model & Background & Non-Alarm Excitation (self/cross) & Alarm Excitation (self/cross) \\
\hline
IPP & 100\% & 0\% & 0\% \\
UHP & 87\% & 13\% (100\%\ / 0\%) & 0\% \\
MHP & 86\% & 15\% (79\%\ / 21\%) & 0\% \\
MHP-R & 85\% & 15\% (79\%\ / 21\%) & 0\% \\
MHPA & 85\% & 13\% (80\%\ / 20\%) & 2\% (26\%\ / 74\%)\\
MHPA-R & 87\% & 11\% (80\%\ / 20\%) & 2\% (25\%\ / 75\%)\\
\hline
\end{tabular}
\caption{Average probability that each event was caused by different parts of the model.}
\label{tab: triggering probabilities}
\end{table}

We perform one further test based on the predictive performance of the best performing model MHPA-R. A difficulty we face when simulating the MHPA-R model is that we must model rate at which events cause an alarm to be sounded. We do not assume any distribution for the alarm in our model, rather we treat it as a covariate. This is not a problem when fitting the model in our training set since the timings of the alarms are known when they are considered in the likelihood function. To overcome this issue, for each event in the test set we simulate ten times from the MHPA-R model conditioned on all previous events, including the previous events in the test set, up until the next arrival. Thus obtaining an interarrival time without needing to directly model alarm occurrences. We then calculate the proportion of interarrival times in four-hour bins and produce a histogram representing the distribution of interarrival times as shown in Figure \ref{fig:Interarrival hist}. The orange bars in Figure \ref{fig:Interarrival hist} represent 95\% error bars which give two standard deviations from the mean proportion of events occurring in the four-hour windows over the ten simulations. We also present the observed histogram of the arrival times over the test period model by the black dots in the relevant interarrival time bin. As we can see, our model seems to capture quite well the observed interarrival densities with all bar one (20-24 hour bar) of the observed densities falling inside the constructed confidence intervals.

\begin{figure}
\centering
\includegraphics[width=14cm]{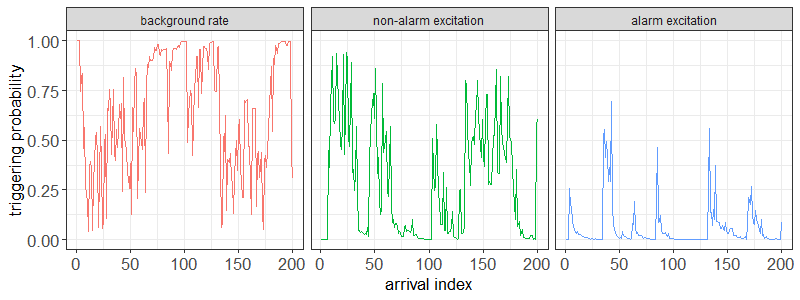}
\caption{Different sources of excitation in the period between the first 200 events. Left plots give probabilities that each event was caused by the background rate, middle plots give probability that events were caused by events that had not sounded an alarm and the right plots give the probability that the event was caused by an alarm. These probabilities were calculated under the model MHPA-R. }
\label{fig:proportion graphs}
\end{figure}

\begin{figure}
\centering
\includegraphics[width=12cm]{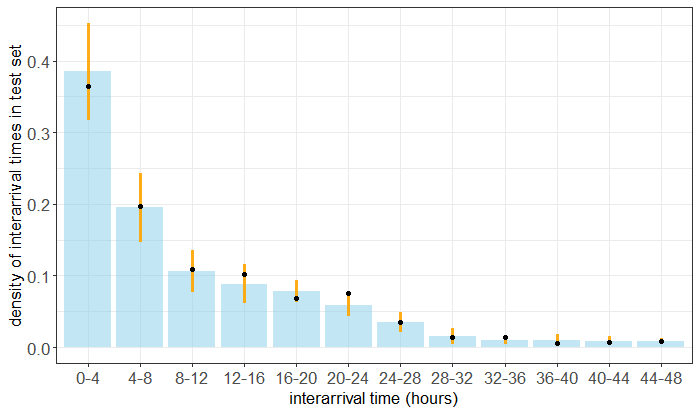}
\caption{Histogram representing the distribution of interarrival times over 10 simulations at each point in the test data (under the MHPA-R model) with 95\% error bars (orange bars). The black dots represents the observed histogram of the arrival times over the test period.}
\label{fig:Interarrival hist}
\end{figure}

\subsection{Interpretation of Results}
We now analyse the parameter fits for the MHPA-R model. Figure \ref{fig:mhp k} and Figure \ref{fig:MHPA k and alpha} present heatmaps representing the values of the excitation matrix in the MHP model and the excitation matrices $K$ and $\alpha$ in the MHPA-R model. That is, the magnitude of the entry  $K_{i,j}$ in the excitation matrix $K$ represents the average number of 'direct offspring' in ward $j$ from an non-alarm sounding event in ward $i$. That is, the number of events in ward $j$ directly caused by an earlier non-alarm sounding event in ward $i$. Further events would then be caused by the direct offspring due to the self exciting nature of the process. The magnitude of entry $\alpha_{i,j}$ in the excitation matrix $\alpha$ represents the average number of direct offspring in ward $j$ from an alarm sounding event in ward $i$. Comparing the $K$-matrices for both models, it can be seen that the shapes of the heatmaps are very similar with the MHPA-R $K$-matrix being scaled down compared to the MHP $K$-matrix. Both matrices are heavily concentrated around the diagonals, implying that a lot of excitation occurs within the same ward. That is, when an event occurs in the same ward and does not cause an alarm to be sounded, further events are caused within the same ward (with lesser events being caused in other wards). Some cross-excitation can also be seen in both matrices and are very similarly structured. Looking at the $\alpha$-matrix in the MHPA-R model, we can see a much more spread out matrix where there is less excitation concentrated along the diagonal, as expected. We expect that the excitation matrix when an alarm has been sounded to spread over the entire hospital. Furthermore, we can see that there is some evidence of excitation due to an alarm in the same ward. As mentioned in previous sections, we would expect that the action of witnessing an event which has caused an alarm to be sounded to have an effect on the rate of future events. In Figure \ref{fig:background}, we analyse the estimated parameters in the rest of our model. Namely, the background rate and the excitation kernel. Figure \ref{fig:background} reports the magnitude of the fitted background constants and the fitted excitation kernel. An interesting result to note here is the mean of our excitation kernel. We could call this the expected 'memory' of the patients throughout the hospital. From the graph we can see that this expected memory is about 1.7 days, which is quite intuitive result when we consider the period in which a patient might have the past event on their mind.

\begin{figure}
\centering
\includegraphics[width=6.5cm]{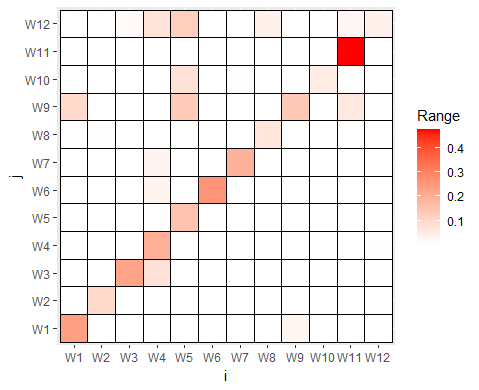}
\caption{Heatmap representing the values of the excitation-matrix, $K$, under the MHP-R model}
\label{fig:mhp k}
\end{figure}

\begin{figure}
    \centering
    \subfloat[\centering $K$]{{\includegraphics[width=6.5cm]{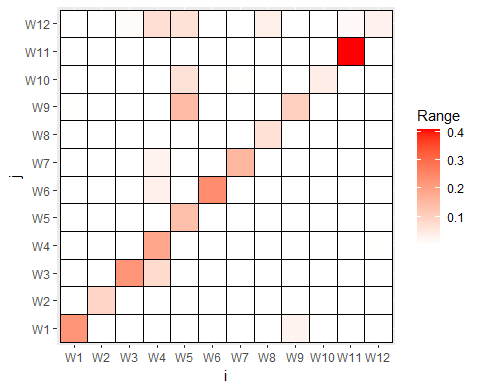} }}%
    \,
    \subfloat[\centering $\alpha$ ]{{\includegraphics[width=6.5cm]{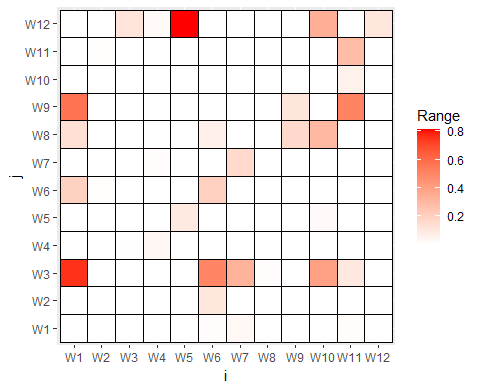} }}%
    \caption{Heatmap representing the values of the self- and cross-excitations from the excitation matrices $K$ and $\alpha$ under the MHPA-R model.}
    \label{fig:MHPA k and alpha}
\end{figure}

\begin{figure}
    \centering
    \subfloat[\centering Plot of fitted background parameters, $\mu^{(m)}$, for each ward $m$. The background rate is proportional to $\mu^{(m)}$ multiplied by the histogram as in Figure \ref{fig:prelim season}.]{{\includegraphics[width=7cm]{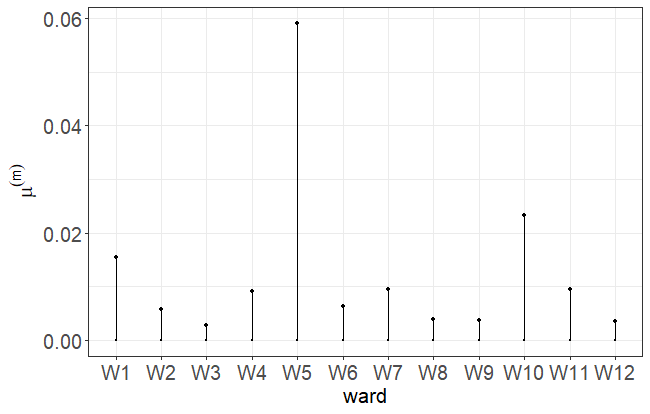} }}%
    \,
    \subfloat[\centering Plot of fitted geometric excitation kernel, $g(t)$. Dotted line represents the expected value, $1/\beta$.]{{\includegraphics[width=7cm]{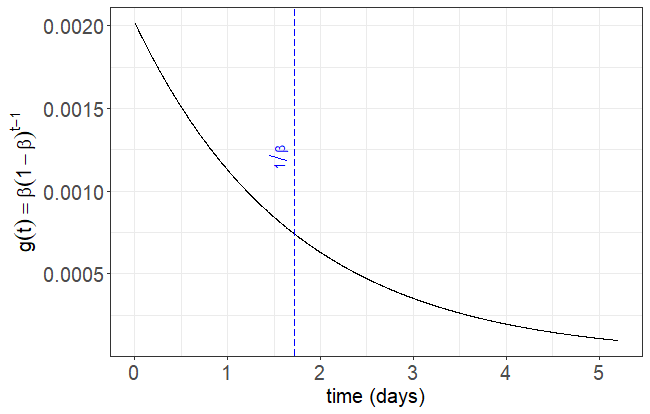} }}%
    \caption{Figures visualising background parameters (scaling constants) and the estimated excitation kernel for MHPA-R model.}
    \label{fig:background}
\end{figure}

\subsection{Forecasting}\label{Forecasting}

In Figure \ref{fig:forecasts}, predictions for the total number of instances over a period of five weeks of the unseen test set, $\textbf{T}_{test}$ are presented. We implement the algorithm developed in Section \ref{Forecasting}, where we estimate the probability of any given event sounding an alarm to be $p=13.66\%$. We do this by assuming a Bernoulli distribution for the alarm generating process and finding the maximum likelihood estimator, the sample mean, from the training data. That is, we set $p$ to be the proportion of events that triggered an alarm in the training data. The filled shapes represent the distributions of one hundred forecasts using the MHPA-R model given the timing and locations of all the previous events. The dotted line represents the sample mean of these forecasts and the solid line indicates the observed number of events over the period. For all wards the true number of events fall within $95\%$ of our forecasts. That is, the observed events fall within the reasonable region of what our model predicts. This is a good indication of a well suited model. The ability to accurately forecast these violent event, with a level of uncertainty quantification, would allow for better organisation of staff within the hospital, which would hopefully help to reduce the number of events in the future.

\begin{figure}
\centering
\includegraphics[width=15cm]{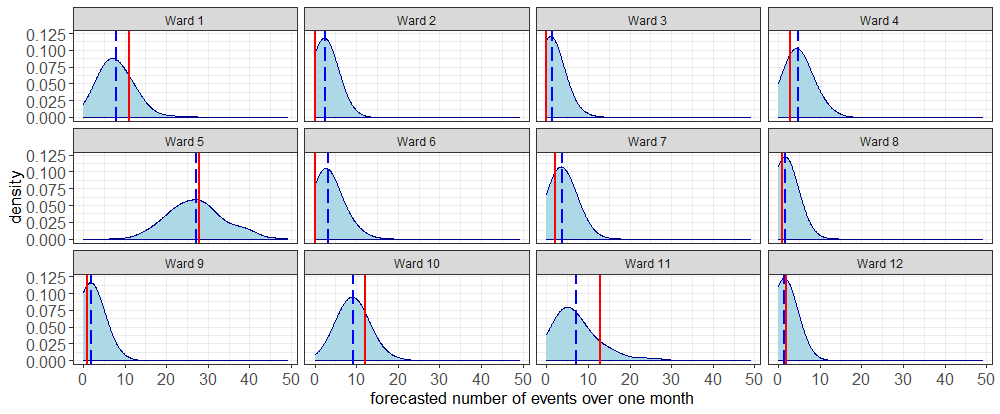}
\caption{Predictive distributions for the number of events over a one-month period for each ward (MHPA-R model). The blue dotted line represents average over one hundred simulations and the red line represents the observed number of events over the interval. The constant probability that an event causes an alarm to be sounded is calculated to be $p=0.1366$.}
\label{fig:forecasts}
\end{figure}

\section{Conclusion}\label{Conclusion}

In this work we tackled two problems. Firstly, in order to apply a discrete Hawkes model to a large dataset, we developed an efficient estimation procedure. We tackled this problem by using a geometric distribution for the excitation kernel and maintained an additional data structure which took a constant space for both the likelihood function and its gradient. Next, we introduced a marked discrete Hawkes process model for modelling the occurrences of violent events within a forensic psychiatric hospital. Our model specification is interpretable and allows for a deeper understanding of the extent to which self-excitation and cross-excitation between wards influence future events throughout the hospital and the role that the alarm system plays as a covariate. We demonstrated that our model offers a significant improvement in the predictive performance when compared to other standard models. This model has important implications in a logistical sense. Firstly, it allows for accurate predictions of the numbers of events within each ward. This could allow for protocol changes within the forensic hospital which may reduce rates of violence in areas with a high number of predicted violent events, such as staff make up, treatment programs and reducing the visibility of the alarm for the patients. Developing an understanding of the the rates of violence and aggression would be useful in improving efficacy of staffing and hopefully lead to an overall reduction in rates of violence. The model reveals a strong excitation component to the events, both with and without an alarm being sounded, which could warrant further investigation into potential underlying factors that could explain the observed excitation. From a psychological point of view, what the excitation effect is likely describing is that being made aware of another incident of violence and aggression might have a triggering effect on other patients through a possible social contagion effect, as described in \cite{Beck}. It is also important to note that the majority of people within forensic psychiatric hospitals are likely to have experienced a lot of trauma. Thus, one would expect them to be hyper vigilant for danger and more sensitive to potential cues of danger, such as the alarm system. This work could be extended in some interesting directions. Firstly, we could relax the condition that the excitation kernel is 1) different for an alarm with and without an event and 2) different between wards. Our approach specified a both of these kernels to have the same shape. The reason we do this is two-fold: To maintain interpretability in our model, which was the primary motivation for this analysis, and identifiability issues. Another interesting direction is to include the type of event into the model. The staff recorded the type of event which took place, such as physical violence or verbal abuse. This seems like a natural extension to our model, but also introduces identifiability issues.\newpage

\bibliographystyle{plainnat}  
\bibliography{templateArxiv}

\newpage
\begin{appendices}
\section{Appendix}

\subsection{Proof of Proposition \ref{First likelihood prop}}\label{First likelihood prop - proof}
We have the following

\begin{align*}
   \sum^{N_{\delta}}_{t=1}\lambda^{(m)}(t) - \sum^{N_{\delta}}_{t=1}\mu^{(m)}(t)  &= \sum^{N_{\delta}}_{t=1}\sum^M_{l=1}\sum_{i:t_{l,i}\leq t-1}Y_{t_{l,i}}^{(l)} f_{l,m}(m(t_{l,i})) g_{l,m}(t-t_{l,i}) \\
   &= \sum^M_{l=1}\sum^{N_{\delta}}_{t=1} \sum^{N_l}_{n=1} Y_{t_n}^{(l)} f_{l,m}(m(t_{l,n})) g_{l,m}(t-t_{l,n})\mathbbm{1}(t_{l,n}<t) \\
   &= \sum^M_{l=1}\sum^{N_l}_{n=1} Y_{t_n}^{(l)} f_{l,m}(m(t_{l,n})) \sum^{N_{\delta}}_{t=1}  g_{l,m}(t-t_n)\mathbbm{1}(t_n<t) \\
   &= \sum^M_{l=1}\sum^{N_l}_{n=1} Y_{t_n}^{(l)} f_{l,m}(m(t_{l,n})) G_{l,m}(N_{\delta}-t_n) \\
   &= \sum^M_{l=1}\sum_{t\in \tau^{(l)}} Y_{t}^{(l)} f_{l,m}(m(t)) G_{l,m}(N_{\delta}-t)
\end{align*}

Finally, summing over $m=1,\ldots,M$ gives the required result.

\subsection{Proof of Corollary \ref{corollary}}\label{corollary - proof}

Applying Proposition \ref{First likelihood prop}, we obtain

\begin{multline*}
    \log L(\theta|\tau,A) = \sum^M_{m=1} \sum_{t\in \tau^{(m)}}  Y^{(m)}_t\log(\lambda^{(m)}(t)) \\
    - \sum^M_{m=1}\sum^M_{l=1} \sum_{t\in \tau^{(l)}} Y_t^{(l)} \Big( K_{l,m} + \alpha_{l,m} \mathbbm{1}(m(t_{l,i})) \Big) (1-(1-\beta)^{N_{\delta}-t})
    - \sum^M_{m=1} \mu_0^{(m)} \sum^{N_{\delta}}_{t=1} s(h(t)).
\end{multline*}

We can expand the bracket and use our notation $A^{(l)}$ as discussed previously to obtain the required result.

\subsection{Proof of Proposition \ref{Recursion 1}}\label{Recursion 1 - proof}

Note that for $j\in \{1,...,N_m-1\}$,

\begin{multline}
    \sum_{i:t_{l,i}\leq t_{lj+1}-1}Y^{(l)}_{t_{l,i}} \beta(1-\beta)^{t_{l,j+1}-t_{l,i}-1} = (1-\beta)^{t_{l,j+1}-t_{l,j}} \sum_{i:t_{l,i}\leq t_{l,j}-1} Y^{(l)}_{t_{l,i}} \beta(1-\beta)^{t_{l,j}-t_{l,i}-1} \\ 
    + Y^{(l)}_{t_{l,j}}\beta(1-\beta)^{t_{l,j+1}-t_{l,j}-1}, 
\end{multline}

and 

\begin{multline}
    \sum_{i:a_{l,i}\leq t_{l,j+1}-1}Y^{(l)}_{t
    a_{l,i}} \beta(1-\beta)^{t_{l,j+1}-a_{l,i}-1} = (1-\beta)^{t_{l,j+1}-t_{l,j}} \sum_{i:a_{l,i}\leq t_j-1} Y^{(l)}_{a_{l,i}} \beta(1-\beta)^{t_{l,j}-a_{l,i}-1} \\ 
    + Y^{(l)}_{t_{l,j}}\mathbbm{1}(t_{j,l}\in A^{(l)})\beta(1-\beta)^{t_{l,j+1}-t_{l,j}-1}, 
\end{multline}

$\forall l \in \{1,...,M\}, j=1,...,N-1$. Then, if we define,

\begin{equation}
    R(j,l):=\sum_{i:t_{l,i}\leq t_{l,j}}Y^{(l)}_{t_{l,i}} \beta(1-\beta)^{t_{l,j}-t_{l,i}-1},
\end{equation}

and 

\begin{equation}
    R_a(j,l):=\sum_{i:a_{l,i}\leq t_j-1} Y^{(l)}_{a_{l,i}} \beta(1-\beta)^{t_{l,j}-a_{l,i}-1},
\end{equation}\label{Ra}

then we obtain the required result.

\begin{lem}\label{partials - alarm}

For the log-likelihood function with conditional intensity function as defined in Equation \eqref{Alarm lambda}, we have the following;

\begin{equation}\label{mu gradient}
    \frac{\partial}{\partial \mu_0^{(p)}} \log L(\theta|\tau,A) = \sum_{t \in \tau^{(p)}} \frac{Y_t^{(p)}s(h(t))}{\lambda^{(p)}(t)} - \sum^{N_{\delta}}_{t=1} s(h(t)),
\end{equation}

\begin{equation}\label{k gradient}
    \frac{\partial}{\partial K_{p,q}} \log L(\theta|\tau,A) = \sum_{t \in \tau^{(q)}} \frac{Y_t^{(q)}}{\lambda^{(q)}(t)} \sum_{i:t_{p,i}\leq t-1}Y^{(p)}_{t_{p,i}} g(t-t_{p,i}) - \sum_{t\in \tau^{(p)}} Y_t^{(p)}G(N_{\delta}-t),
\end{equation}

\begin{equation}\label{alpha gradient}
    \frac{\partial}{\partial \alpha_{p,q}} \log L(\theta|\tau,A) = \sum_{t \in \tau^{(q)}} \frac{Y_t^{(q)}}{\lambda_t^{(q)}} \sum_{i:a_{p,i}\leq t-1}Y_{a_{p,i}}^{(p)}g(t-a_{p,i})- \sum_{a\in A^{(p)}}Y_a^{(p)}G(N_{\delta}-a),
\end{equation}

\begin{multline}\label{beta gradient}
    \frac{\partial}{\partial \beta} \log L(\theta|\tau,A)
    = \sum_{t \in \tau^{(m)}} \sum_{m=1}^{M} \frac{Y_t^{(m)}}{\lambda^{(m)}(t)} \Bigg( \sum_{l=1}^{M} K_{l,m} \sum_{i:t_{l,i}\leq t-1} Y_{t_{l,i}}^{(l)} \frac{\partial}{\partial \beta} g(t-t_{l,i}) \\
    + \sum_{l=1}^{M} \alpha_{l,m} \sum_{i:a_{l,i}\leq t-1} Y_{a_{l,i}}^{(l)} \frac{\partial}{\partial \beta} g(t-a_{l,i})\Bigg) \\
    - \sum_{m=1}^{M} \sum_{l=1}^{M} \sum_{t\in \tau^{(l)}} K_{l,m} Y^{(l)}_t \frac{\partial}{\partial \beta}G(N_{\delta}-t) -  \sum_{m=1}^{M} \sum_{l=1}^{M} \sum_{a\in A^{(l)}} \alpha_{l,m} Y^{(l)}_t \frac{\partial}{\partial \beta}G(N_{\delta}-a), 
\end{multline}

where $g(t)=\beta(1-\beta)^{t-1}$, $G(t)=1-(1-\beta)^t$.  

\end{lem}
    
\begin{proof}
    By application of the chain rule.
\end{proof}

\subsection{Proof of Proposition \ref{corr 1}}\label{corr 1 - proof}

For the first three sets of partial derivatives, we can simply replace replace $\sum_{i:t_{p,i}\leq t-1}Y^{(p)}_{t_{p,i}} g(t-t_{p,i})$ and $\sum_{i:a_{p,i}\leq t-1}Y_{a_{p,i}}^{(p)}g(t-a_{p,i})$ in Lemma \ref{partials - alarm} by their recursive definitions as defined in the proof of Proposition \ref{Recursion 1}, which then allows us to write $\lambda^{(m)}(t)$ in its recursive form. Now, note that we can write Equation \eqref{beta gradient} as 

\begin{align*} 
\frac{\partial}{\partial \beta} G(N_{\delta}-t) &= \frac{\partial}{\partial \beta} \Big( 1- (1-\beta)^{N_{\delta}-t} \Big) \\ 
&=  (N_{\delta}-t)(1-\beta)^{N_{\delta}-t-1},
\end{align*}

and 

\begin{align*}
    \frac{\partial}{\partial \beta} g(t-t_{l,i}) &= \frac{\partial}{\partial \beta} \Big( \beta(1-\beta)^{t-t_{l,i}-1} \Big) \\
   &=  (1-\beta)^{t-t_{l,i}-2}(1-\beta(t-t_{l,i})).
\end{align*}

Hence we can write

\begin{multline}
    \frac{\partial}{\partial \beta} \log L(\theta|\tau,A) \\ 
    = \sum_{t \in \tau^{(m)}} \sum_{m=1}^{M} \frac{Y_t^{(m)}}{\lambda^{(m)}(t)} \Bigg( \sum_{l=1}^{M} K_{l,m} \sum_{i:t_{l,i}\leq t-1} Y_{t_{l,i}}^{(l)} (1-\beta)^{t-t_{l,i}-2}(1-\beta(t-t_{l,i})) \\
    + \sum_{l=1}^{M} \alpha_{l,m} \sum_{i:a_{l,i}\leq t-1} Y_{a_{l,i}}^{(l)} (1-\beta)^{t-a_{l,i}-2}(1-\beta(t-a_{l,i}))\Bigg) \\
    - \sum_{m=1}^{M} \sum_{l=1}^{M} \sum_{t\in \tau^{(l)}} K_{l,m} Y^{(l)}_t \frac{\partial}{\partial \beta}G(N_{\delta}-t) -  \sum_{m=1}^{M} \sum_{l=1}^{M} \sum_{a\in A^{(l)}} \alpha_{l,m} Y^{(l)}_t \frac{\partial}{\partial \beta}G(N_{\delta}-a). 
\end{multline}

Now define

\begin{align}
    R^1(j,l)&:=\sum_{i:t_{l,i}\leq t-1} Y_{t_{l,i}}^{(l)} (1-\beta)^{t-t_{l,i}-2}, \\ 
    R^1_a(j,l)&:= \sum_{i:a_{l,i}\leq t-1} Y_{a_{l,i}}^{(l)} (1-\beta)^{t-a_{l,i}-2}, \\ 
    R^{2}(j,l)&:=\sum_{i:t_{l,i}\leq t-1} Y_{t_{l,i}}^{(l)} (1-\beta)^{t-t_{l,i}-2}\beta(t-t_{l,i}), \\
    R^{2}_a(j,l)&:=\sum_{i:a_{l,i}\leq t-1} Y_{a_{l,i}}^{(l)} (1-\beta)^{t-a_{l,i}-2}\beta(t-a_{l,i}).
\end{align}

Plugging these sums into our recursive relations then give the required result.

\subsection{Proof of Lemma \ref{R3 - with alarm}}\label{R3 - with alarm - proof}

Define $R^3(t,l):=\sum_{i:t_{l,i}\leq t-1}{Y^{(l)}_{t_{l,i}} \beta(1-\beta)^{t-t_{l,i}-1}}, \forall t=1,\ldots,N_{\delta}$. Then, we have that for $t=1,\ldots,N_{\delta}-1$,

\begin{align*}
    R^3(t+1,l)&=\sum_{i:t_i\leq t}{Y^{(l)}_{t_{l,i}} \beta_{l,m}(1-\beta)^{t-t_{l,i}}} \\
    &=(1-\beta)\sum_{i:t_{l,i}\leq t-1}{Y^{(l)}_{t_{l,i}} \beta(1-\beta)^{t-t_{l,i}-1}} + Y^{(l)}_t\beta\\
    &=(1-\beta)R^3(t,l) + Y^{(l)}_t\beta.
\end{align*}

Next define $R^3_a(t,l):=\sum_{i:a_{l,i}\leq t-1} Y^{(l)}_{a_{l,i}} \beta(1-\beta)^{t-a_{l,i}-1}$, $\forall t=1,\ldots,N_{\delta}$. Then, we have that for $t=1,\ldots,N_{\delta}-1$,

\begin{align*}
    R^3_a(t+1,l)&=\sum_{i:a_{l,i} \leq t}{Y^{(l)}_{a_{l,i}} \beta(1-\beta)^{t-a_{l,i}}} \\
    &=(1-\beta)\sum_{i:a_{l,i} \leq t-1}{Y^{(l)}_{a_{l,i}} \beta(1-\beta)^{t-a_{l,i}-1}} + Y^{(l)}_t\mathbbm{1}(t\in A^{(l)})\beta \\
    &=(1-\beta)R_a^3(t,l) + Y^{(l)}_t\mathbbm{1}(t\in A^{(l)})\beta.
\end{align*}

Finally, by noting $R^3(1,l)=R_a^3(1,l)=0$ we obtain the required result.

\section{Estimation of Multivariate Discrete Hawkes Processes}\label{Hawkes estimation}
We present the results from the main body of this paper for the unmarked discrete Hawkes process. We believe this process could be more widely applicable and thus worth including. 

Consider an unmarked multivariate discrete Hawkes process with intensity function defined as follows

\begin{equation}\label{Hawkes lambda}
    \lambda^{(m)}(t) = \mu^{(m)} + \sum^M_{l=1} K_{l,m} \sum_{i:t_{l,i}\leq t-1} Y^{(l)}_{t_{l,i}}  \beta_{l,m}(1-\beta_{l,m})^{t-t_{l,i}-1}, 
\end{equation}

where $\mu^{(m)}>0$, $\beta_{l,m}>0$ and $K_{l,m}>0$. 

\begin{corollary}\label{corollary appendix}
For the intensity function as defined by Equation \eqref{Hawkes lambda}, by applying Proposition \ref{First likelihood prop}, we can rewrite the log-likelihood function as follows

\begin{equation}
    \log L(\theta|\tau) = \sum^M_{m=1}\sum_{t\in \tau^{(m)}} Y_t^{(m)} \log(\lambda^{(m)}(t)) - \sum_{m=1}^M \sum_{l=1}^M K_{l,m} \sum_{t\in \tau^{(l)}} Y_t^{(l)} (1-(1-\beta_{l,m})^{N_{\delta}-t}) - N_{\delta}\mu^{(m)} 
\end{equation}

\end{corollary}

\begin{proof}
    By direct application of Proposition \ref{First likelihood prop}.
\end{proof}

\begin{prop}\label{Recursion 1- no alarm} For $j\in \{1,\ldots,N_m\}$ and $m \in \{1,\ldots,M\}$, the conditional intensity function as in Equation \eqref{Hawkes lambda} recursively as 

\begin{equation}
    \lambda^{(m)}(t_{l,j})= \mu^{(m)} + \sum^M_{l=1}{K_{l,m}R(j,l,m)}
\end{equation}

where $R$ is defined as follows: $\forall l \in \{1,...,M\}$, $j\in \{1,\ldots,N_m-1\}$,

\begin{equation}\label{R no alarm}
    R(j+1,l,m) := (1-\beta_{l,m})^{t_{l,j+1}-t_{l,j}}R(j,l,m)+Y_{t_{l,j}}^{(l)}\beta_{l,m}(1-\beta_{l,m})^{t_{l,j+1}-t_{l,j}-1}
\end{equation}

and we define $R(1,l,m):=0$.

\end{prop}

\begin{proof}

Note that for $j\in \{1,\ldots,N_m-1\}$,

\begin{multline}
    \sum_{i:t_{l,i}<t_{l,j+1}}y^{(l)}_{t_{l,i}} \beta_{l,m}(1-\beta_{l,m})^{t_{l,j+1}-t_{l,i}-1} 
    \\ = (1-\beta_{l,m})^{t_{l,j+1}-t_{l,j}} \sum_{i:t_{l,i}<t_{l,j}} y^{(l)}_{t_{l,i}} \beta_{l,m}(1-\beta_{l,m})^{t_{l,j}-t_{l,i}-1} 
    + y^{(l)}_{t_{l,j}}\beta_{l,m}(1-\beta_{l,m})^{t_{l,j+1}-t_{l,j}-1}, 
\end{multline}

Then, if we define,

\begin{equation}
    R(j,l,m):=\sum_{i:t_{l,i}<t_{j}}y^{(l)}_{t_{l,i}} \beta_{l,m}(1-\beta_{l,m})^{t_{l,j}-t_{l,i}-1},
\end{equation}

we obtain the required result.

\end{proof}

\begin{lem}\label{partials - no alarm}
For the intensity function as in Equation \eqref{Hawkes lambda}, the gradient of the log-likelihood function is given by the following partial derivatives;

\begin{equation}
    \frac{\partial}{\partial \mu^{(p)}} \log L(\theta|\tau) = \sum_{t \in \tau^{(p)}} \frac{Y_t^{(p)}}{\lambda_t^{(p)}} - N_{\delta}
\end{equation}

\begin{equation}
    \frac{\partial}{\partial K_{p,q}} \log L(\theta|\tau) = \sum_{t \in \tau^{(q)}} \frac{Y_t^{(q)}}{\lambda_t^{(q)}} \sum_{i:t_{p,i}<t}Y_{t_{p,i}}^{(p)}g_{p,q}(t-t_{p,i}) - \sum_{t\in \tau^{(p)}}Y_t^{(p)}G_{p,q}(N_{\delta}-t)
\end{equation}

\begin{equation}\label{beta gradient - no alarm}
    \frac{\partial}{\partial \beta_{p,q}} \log L(\theta|\tau)
    = \sum_{t \in \tau^{(q)}} \frac{Y_t^{(q)}}{\lambda_t^{(q)}} K_{p,q} \sum_{i:t_{p,i}<t} Y_{t_{p,i}}^{(p)} \frac{\partial}{\partial \beta_{p,q}} g_{p,q}(t-t_{p,i}) 
    -\sum_{t\in \tau^{(p)}} k_{p,q} Y^{(p)}_t \frac{\partial}{\partial \beta_{p,q}}G_{p,q}(N_{\delta}-t),
\end{equation}

where $g(t)=\beta_{l,m}(1-\beta_{l,m})^{t-1}$ and $G(t)$ is the cumulative distribution function for $g(t)$.

\end{lem}

\begin{proof}
    By application of the chain rule.
\end{proof}

\begin{prop}\label{corr 1 - no alarm}
If we define the intensity function as in Equation \eqref{Hawkes lambda}, we can write the gradient of the log-likelihood function as follows; for $j\in \{1,\ldots,N_l\}$ and $l \in \{1,\ldots,M\}$,

\begin{equation}
    \frac{\partial}{\partial \mu^{(p)}} \log L(\theta|\tau) = \sum_{j:t_{p,j} \in \tau^{(p)}} \frac{Y_{t_{p,j}}^{(p)}}{\mu^{(p)} + \sum_{l=1}^M k_{l,p} R(j,l,p)} - N_{\delta}
\end{equation}

\begin{equation}
    \frac{\partial}{\partial K_{p,q}} \log L(\theta|\tau) = \sum_{j:t_{q,j} \in \tau^{(q)}} \frac{Y_{t_{q,j}}^{(q)} R(j,p,q)}{\mu^{(q)} + \sum_{l=1}^M k_{l,q} R(j,l,q)} 
    -\sum_{t\in \tau^{(p)}}Y_t^{(p)}(1-(1-\beta_{p,q})^{N_{\delta}-t})
\end{equation}

\begin{multline}
    \frac{\partial}{\partial \beta_{p,q}} \log L(\theta|\tau)
     = \sum_{j:t_{q,j} \in \tau^{(q)}} \frac{Y_{t_{q,j}}^{(q)}}{\mu^{(q)} + \sum_{l=1}^M k_{lm} R(j,l,q)} K_{p,q} \Big( R^{1}(j,p,q)-R^{2}(j,p,q)\Big) \\   
    -\sum_{t\in \tau^{(p)}} k_{p,q} Y^{(p)}_t (N_{\delta}-t)(1-\beta_{p,q})^{N_{\delta}-t-1}.
\end{multline}
$R$ is defined as in Proposition \ref{Recursion 1- no alarm} and $R^1$, $R^2$ are defined as follows: $R^{1}(1,l,m):=R^{2}(1,l,m):=0$ $\forall l,m \in {1,...,M}$ and we define recursively, for $j\in \{1,\ldots,N_l\}$, 

\begin{equation}
    R^{1}(j+1,l,m) := (1-\beta_{l,m})^{t_{l,j+1}-t_{l,j}}R^{1}(j,l,m)+Y_{t_{l,j}}^{(l)}(1-\beta_{l,m})^{t_{l,j+1}-t_{l,j}-2},
\end{equation}
 and 
\begin{equation}
    R^{2}(j+1,l,m) := \beta_{l,m}(t_{l,j+1}-t_{l,j})R^{1}(j+1,l,m)+(1-\beta_{l,m})^{t_{l,j+1}-t_{l,j}}R^2(j,l,m). 
\end{equation}

\end{prop}

\begin{proof}
     
For the first two sets of partial derivatives, we can simply replace  replace $\sum_{i:t_{l,i}<t}Y_{t_{l,i}}^{(p)}g_{p,q}(t-t_{l,i})$ in Lemma \ref{partials - no alarm} by its recursive definition as defined in the proof of Proposition \ref{Recursion 1- no alarm}, and write the intensity function in its recursive form as defined in Proposition \ref{Recursion 1- no alarm}. Now, note that 

\begin{align*} 
\frac{\partial}{\partial \beta_{l,m}} G(N_{\delta}-t) &= \frac{\partial}{\partial \beta_{l,m}} \Big( 1- (1-\beta_{l,m})^{N_{\delta}-t} \Big) \\ 
&=  (N_{\delta}-t)(1-\beta_{l,m})^{N_{\delta}-t-1}
\end{align*}

and 

\begin{align*}
    \frac{\partial}{\partial \beta_{l,m}} g(t-t_{l,i}) &= \frac{\partial}{\partial \beta_{l,m}} \Big( \beta_{l,m}(1-\beta_{l,m})^{t-t_{l,i}-1} \Big) \\
   &=  (1-\beta_{l,m})^{t-t_{l,i}-2}(1-\beta_{l,m}(t-t_{l,i})).
\end{align*}

Hence, we can write Equation \eqref{beta gradient - no alarm} as

\begin{multline}\label{beta gradient sub - no alarm}
    \frac{\partial}{\partial \beta_{p,q}} \log L(\theta|\tau)
    = \sum_{t \in \tau^{(q)}}  \frac{Y_t^{(q)}}{\lambda_t^{(q)}}  K_{p,q} \sum_{i:t_{l,i}<t} Y_{t_{l,i}}^{(p)} (1-\beta_{p,q})^{t-t_{l,i}-2}(1-\beta_{l,m}(t-t_{l,i}))  \\
    -\sum_{t\in \tau^{(p)}}  k_{p,q} Y^{(p)}_t (N_{\delta}-t)(1-\beta_{l,m})^{\textbf{}-t-1}
\end{multline}

Now, define 
\begin{align}
    R^1(j,l,m)&:=\sum_{i:t_{l,i}<t_{l,j}}Y_{t_{l,i}}^{(l)}(1-\beta_{l,m})^{t_{l,j}-t_{l,i}-2} \\ 
    R^{2}(j,l,m)&:=\sum_{i:t_{l,i}<t_{l,j}}Y_{t_{l,i}}^{(l)}(1-\beta_{l,m})^{t_{l,j}-t_{l,i}-2}\beta_{l,m}(t_{l,j}-t_{l,i}). 
\end{align}

Then, noting that $R^1(1,l):=R^2(1,l)=0$, we have that, 

\begin{align*}
    R^1(j+1,l,m)&=\sum_{i:t_{l,i}<t_{j+1}}Y_{t_{l,i}}^{(l)}(1-\beta_{l,m})^{t_{j+1}-t_{l,i}-2} \\ 
    &=(1-\beta_{l,m})^{t_{j+1}-t_{l,j}}\sum_{i:t_{l,i}<t_{l,j}}Y_{t_{l,i}}^{(l)}(1-\beta_{l,m})^{t_{l,j}-t_{l,i}-2} + Y_{t_{l,j}}^{(l)}(1-\beta_{l,m})^{t_{j+1}-t_{l,j}-2} \\
    &=(1-\beta_{l,m})^{t_{j+1}-t_{l,j}}R^1(j,1,m) + Y_{t_{l,j}}^{(l)}(1-\beta_{l,m})^{t_{j+1}-t_{l,j}-2}
\end{align*}

and that 

\begin{align*}
    R^{2}(j+1,l,m)&:=\sum_{i:t_{l,i}<t_{l,j}}Y_{t_{l,i}}^{(l)}(1-\beta_{l,m})^{t_{l,j}-t_{l,i}-2}\beta_{l,m}(t_{l,j}-t_{l,i}) \\ 
    &=(1-\beta_{l,m})^{t_{j+1}-t_{l,j}} \sum_{t_{l,i}<t_{j+1}}Y_{t_{l,i}}^{(l)}(1-\beta_{l,m})^{t_{l,j}-t_{l,i}-2}\beta_{l,m}(t_{j+1}-t_{l,i}) \\
    &=(1-\beta_{l,m})^{t_{j+1}-t_{l,j}} \sum_{t_{l,i}<t_{j+1}}Y_{t_{l,i}}^{(l)}(1-\beta_{l,m})^{t_{l,j}-t_{l,i}-2}\beta_{l,m}(t_{j+1}-t_{l,j}+t_{l,j}-t_{l,i}) \\ 
    &=(1-\beta_{l,m})^{t_{j+1}-t_{l,j}} \sum_{t_{l,i}<t_{j+1}}Y_{t_{l,i}}^{(l)}(1-\beta_{l,m})^{t_{l,j}-t_{l,i}-2}\beta_{l,m}(t_{l,j}-t_{l,i})\\ 
    &\hspace{1cm}  + (1-\beta_{l,m})^{t_{j+1}-t_{l,j}} \sum_{t_{l,i}<t_{j+1}}Y_{t_{l,i}}^{(l)}(1-\beta_{l,m})^{t_{l,j}-t_{l,i}-2}\beta_{l,m}(t_{j+1}-t_{l,j}) \\ 
    &=(1-\beta_{l,m})^{t_{j+1}-t_{l,j}} \sum_{t_{l,i}<t_{j+1}}Y_{t_{l,i}}^{(l)}(1-\beta_{l,m})^{t_{l,j}-t_{l,i}-2}\beta_{l,m}(t_{l,j}-t_{l,i})\\ 
    &\hspace{1cm}  + \sum_{t_{l,i}<t_{j+1}}Y_{t_{l,i}}^{(l)}(1-\beta_{l,m})^{t_{j+1}-t_{l,i}-2}\beta_{l,m}(t_{j+1}-t_{l,j}) \\ 
    &=(1-\beta_{l,m})^{t_{j+1}-t_{l,j}} \sum_{t_{l,i}<t_{j+1}}Y_{t_{l,i}}^{(l)}(1-\beta_{l,m})^{t_{l,j}-t_{l,i}-2}\beta_{l,m}(t_{l,j}-t_{l,i})\\ 
    &\hspace{1cm}  + \beta_{l,m}(t_{j+1}-t_{l,j}) \sum_{t_{l,i}<t_{j+1}}Y_{t_{l,i}}^{(l)}(1-\beta_{l,m})^{t_{j+1}-t_{l,i}-2} \\ 
    &=(1-\beta_{l,m})^{t_{j+1}-t_{l,j}} \sum_{t_{l,i}<t_{l,j}}Y_{t_{l,i}}^{(l)}(1-\beta_{l,m})^{t_{l,j}-t_{l,i}-2}\beta_{l,m}(t_{l,j}-t_{l,i})\\ 
    &\hspace{1cm} + Y_{t_{l,j}}^{(l)}(1-\beta_{l,m})^{-2}\beta_{l,m}(t_{l,j}-t_{l,j}) \\
    &\hspace{2cm} + \beta_{l,m}(t_{j+1}-t_{l,j}) \sum_{t_{l,i}<t_{j+1}}Y_{t_{l,i}}^{(l)}(1-\beta_{l,m})^{t_{j+1}-t_{l,i}-2} \\ 
    &=(1-\beta_{l,m})^{t_{j+1}-t_{l,j}} \sum_{t_{l,i}<t_{l,j}}Y_{t_{l,i}}^{(l)}(1-\beta_{l,m})^{t_{l,j}-t_{l,i}-2}\beta_{l,m}(t_{l,j}-t_{l,i})\\
    &\hspace{1cm} + \beta_{l,m}(t_{j+1}-t_{l,j}) \sum_{t_{l,i}<t_{j+1}}Y_{t_{l,i}}^{(l)}(1-\beta_{l,m})^{t_{j+1}-t_{l,i}-2} \\
    &=(1-\beta_{l,m})^{t_{j+1}-t_{l,j}} R^2(j,l,m) + \beta_{l,m}(t_{j+1}-t_{l,j}) R^1(j+1,l,m)
\end{align*}

By plugging these recursive relations into Equation \eqref{beta gradient sub - no alarm}, we obtain the required result.

\end{proof}

\section{Forecasting for Multivariate Discrete Hawkes Processes}\label{Hawkes forecasting}

\begin{algorithm}
	\caption{Naive algorithm to simulate from a discrete-time Hawkes process with geometric excitation kernel.}
	\label{alg:Naive discrete Hawkes}
	\begin{algorithmic}[1]
	\State \textbf{Require} $\mu^{(1)},\ldots,\mu^{(M)}, K_{1,1},\ldots,K_{M,M}, \beta_{1,1},\ldots,\beta_{M,M}$
	\State Generate $Y^{(m)}_1 \gets \Poisson(\mu^{(m)})$, for each $m=1,\ldots,M$
		\For {$t=2,\ldots,N_{\delta}$}
			\For {$l=1,2,\ldots,M$}
				\State $\lambda^{(m)}(t) \gets \mu^{(m)} + \sum^M_{l=1}{K_{l,m}\sum_{i:t_{l,i}<t}{Y^{(l)}_{t_{l,i}} \beta_{l,m}(1-\beta_{l,m})^{t-t_{l,i}-1}}}$
				\State Generate $Y^{(m)}_t \gets \Poisson(\lambda^{(m)}(t))$
			\EndFor
		\EndFor
	\end{algorithmic} 
\end{algorithm}

\newpage
\begin{lem}\label{R3}
We can write the intensity function as defined by Equation \eqref{Hawkes lambda}, for each discrete time step $t=1,2,\ldots,N_{\delta}$, as follows 

\begin{equation}
    \lambda^{(m)}(t) = \mu^{(m)}+\sum^M_{l=1}K_{lm}R^3(t,l,m),
\end{equation}

where $\forall l=\{1,...,M\}$ we define $R^3(1,l,m):=0$ and

\begin{equation}
    R^3(t+1,l,m)=(1-\beta_{l,m})R^3(t,l,m)+Y^{(l)}_t\beta_{l,m}, \; \quad \text{for} \hspace{0.1cm} t=1,2,\ldots,N_{\delta}-1. 
\end{equation}

\end{lem}

\begin{proof}
    Define $R^3(t,l,m):=\sum_{i:t_{l,i}<t}{Y^{(l)}_{t_{l,i}} \beta_{l,m}(1-\beta_{l,m})^{t-t_{l,i}-1}}, \forall t=1,\ldots,N_{\delta}$. Then, we have that

\begin{align*}
    R^3(t+1,l,m)&=\sum_{i:t_{l,i}<t+1}{Y^{(l)}_{t_{l,i}} \beta_{l,m}(1-\beta_{l,m})^{t-t_{l,i}}} \\
    &= \sum_{i:t_{l,i}<t}{Y^{(l)}_{t_{l,i}} \beta_{l,m}(1-\beta_{l,m})^{t-t_{l,i}}} + Y^{(l)}_t\beta_{l,m}(1-\beta_{l,m})^t-t \\
    &=\sum_{i:t_{l,i}<t}{Y^{(l)}_{t_{l,i}} \beta_{l,m}(1-\beta_{l,m})^{t-t_{l,i}}} + Y^{(l)}_t\beta_{l,m} \\
    &=(1-\beta_{l,m})\sum_{i:t_{l,i}<t}{Y^{(l)}_{t_{l,i}} \beta_{l,m}(1-\beta_{l,m})^{t-t_{l,i}-1}} + Y^{(l)}_t\beta_{l,m}\\
    &=(1-\beta_{l,m})R^3(t,l,m) + Y^{(l)}_t\beta_{l,m}
\end{align*}

Hence, we obtain 

\begin{equation*}
    \lambda^{(m)}_t = \mu^{(m)}+\sum^M_{l=1}K_{lm}R^3(t,l), \forall t=1,\ldots,N_{\delta}.
\end{equation*}

Finally, noting that $R^3(1,l,m)=0$ $\forall m=\{1,...,M\}$, we obtain the required result.

\end{proof}

\begin{algorithm}
	\caption{Recursive algorithm to simulate from a discrete-time Hawkes process with Geometric excitation kernel.}
	\label{alg:Recursive}
	\begin{algorithmic}[1]
	\State \textbf{Require} $\mu^{(1)},\ldots,\mu^{(M)}, K_{1,1},\ldots,K_{M,M}, \beta_{1,1},\ldots,\beta_{M,M}$
	\State Generate $Y^{(m)}_1 \gets \Poisson(\mu^{(m)})$ for $m=1,\ldots,M$
	\State Set $R^3(1,l,m) \gets 0$, $\forall l,m=\{1,\ldots,M\}$
		\For {$t=2,\ldots,N_{\delta}$}
			\For {$m=1,\ldots,M$}
			\For {$l=1,\ldots,M$}
			    \State $R^3(t,l,m)=(1-\beta_{l,m})R^3(t-1,l,m)+Y^{(l)}_{t-1}\beta_{l,m}$
			    \EndFor
				\State $\lambda^{(m)}(t) \gets \mu^{(m)} + \sum^M_{l=1}{K_{l,m}R^3(t,l,m)}$
				\State Generate $Y^{(m)}_t \gets \Poisson(\lambda^{(m)}(t))$
		 \EndFor
		\EndFor
	\end{algorithmic} 
\end{algorithm}

\end{appendices}

\end{document}